\documentclass[11pt]{article}
\usepackage{times}
\usepackage{amsmath}
\usepackage{amsfonts}
\usepackage{amssymb}
\usepackage{amsthm}
\usepackage{multicol}
\usepackage{fullpage}
\usepackage{graphicx}
\usepackage[vlined,ruled,linesnumbered]{algorithm2e}
\usepackage[normalem]{ulem}

\newtheorem{prop}{Proposition}
\newtheorem{lem}{Lemma}
\newtheorem{thm}{Theorem}
\newtheorem{defi}{Definition}
\newtheorem{coro}[lem]{Corollary}

\def\KSW{{\it KSW }}
\def\sgn{\mathop{\rm sgn}}

\begin{document}

\title{On the Hyperbolicity of Small-World and \\ Tree-Like Random Graphs\thanks{This is 
	the version to appear in the journal of Internet Mathematics.}}

\author{
Wei Chen\\
Microsoft Research Asia\\
weic@microsoft.com
\and
Wenjie Fang\\
Ecole Normale Sup\'erieure de Paris\\
Wenjie.Fang@ens.fr
\and
Guangda Hu\\
Princeton University\\
guangdah@cs.princeton.edu
\and
Michael W. Mahoney\\
Stanford University\\
mmahoney@cs.stanford.edu
}

\date{}

\maketitle



\begin{abstract}
Hyperbolicity is a property of a graph that may be viewed as being a 
``soft'' version of a tree, and recent empirical and theoretical work has 
suggested that many graphs arising in Internet and related data applications
have 
hyperbolic properties.
Here,
we consider Gromov's notion of $\delta$-hyperbolicity, and 
we establish several positive and negative results for small-world and 
tree-like random graph models.
First, we study the hyperbolicity of 
	the class of Kleinberg small-world random graphs $\KSW(n, d, \gamma)$, where
	$n$ is the number of vertices in the graph, $d$ is the dimension of the underlying
	base grid $B$, and $\gamma$ is the small-world parameter such that each node
	$u$ in the graph connects to another node $v$ in the graph with probability
	proportional to $1/d_B(u,v)^\gamma$ with $d_B(u,v)$ being the grid distance from $u$ to $v$
	in the base grid $B$.
We show that when $\gamma = d$, the parameter value allowing efficient decentralized routing in
	Kleinberg's small-world network, with probability $1-o(1)$ the hyperbolic $\delta$
	is $\Omega((\log n)^{\frac{1}{1.5(d+1)+\varepsilon}})$
	for any $\varepsilon>0$ independent of $n$.
Comparing to the diameter of $\Theta(\log n)$ in this case, it indicates that 
	hyperbolicity is not significantly improved comparing to graph diameter even when
	the long-range connections greatly improves decentralized navigation.
We also show that for other values of $\gamma$ the hyperbolic $\delta$ is either
	at the same level or very close to the graph diameter, indicating poor hyperbolicity
	in these graphs as well.
Next we study  a class of tree-like graphs called ringed trees
	that have constant hyperbolicity.
We show that adding random links among the leaves in a manner similar to the
	small-world graph constructions may easily destroy the hyperbolicity
	of the graphs, except for a class of random edges added using
	an exponentially decaying probability function based on the ring distance
	among the leaves.

Our study provides one of the first significant analytical results on the 
	hyperbolicity of a rich class of random graphs, which
	shed light on 
	the relationship between hyperbolicity and 
	navigability of random graphs, as well as on the sensitivity
	of hyperbolic $\delta$ to noises in random graphs.
 \end{abstract}

\section{Introduction}
\label{sxn:intro}

Hyperbolicity, a property of metric spaces that generalizes the idea of 
Riemannian manifolds with negative curvature, has received considerable
attention in both mathematics and computer science.
When applied to graphs, 
one may think of hyperbolicity as characterizing a ``soft'' version of a 
tree---trees have hyperbolicity zero, and graphs 
that ``look like'' trees in terms of their metric structure have ``small'' 
hyperbolicity.
Since trees are an important class of graphs and since tree-like graphs 
arise in numerous applications, the idea of hyperbolicity has received 
attention in a range of applications.
For example, it has found usefulness in the visualization of the Internet, 
the Web, and other large graphs~\cite{LR94,LRP95,MB95,Mun98,WR02};
it has been applied to questions of compact routing, navigation, and 
decentralized search in Internet graphs and small-world social 
networks~\cite{GL05,CDEHVX10,Kle07,ABKMRT07,KCFB07,BKC09,PKBV10};
and it has been applied to a range of other problems such as distance 
estimation, network security, sensor networks, and traffic flow and 
congestion minimization~\cite{Bar02,JL02,JL04,JLB08,NS11,montgolfier2011treewidth}.

The hyperbolicity of graphs is typically measured by Gromov's
	hyperbolic $\delta$~\cite{gromov1987hyperbolic,BH99}
	(see Section~\ref{sec:pre}).
The hyperbolic $\delta$ of a graph measures the ``tree-likeness'' of
	the graph in terms of the graph distance metric.
It can range from $0$ up to the half of the 
graph diameter, with trees having $\delta=0$, in contrast of ``circle graphs'' and 
``grid graphs'' having large $\delta$ equal to roughly half of their diameters.

In this paper, we study the $\delta$-hyperbolicity of 
	families of random graphs that intuitively 
	have some sort of tree-like or hierarchical structure.
Our motivation comes from two angles.
First, although there are a number of empirical studies on the hyperbolicity
	of real-world and random 
	graphs~\cite{Bar02,JL02,LouTH,LohsoonTH,NS11,montgolfier2011treewidth}, there are essentially no systematic analytical study on the
	hyperbolicity of popular random graphs.
Thus, our work is intended to fill this gap.
Second, a number of algorithmic studies show that good graph hyperbolicity
	leads to efficient distance labeling and routing 
	schemes~\cite{CD00,GL05,CE07,chepoi2008diameters,KPKVB10,CDEHVX10},
	and the routing infrastructure of the Internet is also empirically
	shown to be hyperbolic~\cite{Bar02}.
Thus, it is interesting to further investigate if efficient routing capability
	implies good graph hyperbolicity.

To achieve our goal, we first provide fine-grained characterization of
	$\delta$-hyperbolicity of graph families relative to the graph diameter:
A family of random graphs is 
	(a) {\em constantly hyperbolic} if their hyperbolic $\delta$'s
	are constant, regardless of the size or diameter of the graphs; 
	(b) {\em logarithmically (or polylogarithmically) hyperbolic} if their hyperbolic $\delta$'s
	are in the order of logarithm (or polylogarithm) of the graph diameters;
	(c) {\em weakly hyperbolic} if their hyperbolic $\delta$'s grow 
	asymptotically slower than the graph diameters; and
	(d) {\em not hyperbolic} if their hyperbolic $\delta$'s are at
	the same order as the graph diameters.

We study two families of random graphs.
The first family is Kleinberg's grid-based small-world random 
	graphs~\cite{Kle00}, which build random long-range edges among pairs
	of nodes
	with probability inverse proportional to the $\gamma$-th power of the
	grid distance of the pairs.
Kleinberg shows that when $\gamma$ equals to the grid dimension $d$,
the number of hops for
	decentralized routing can be improved from $\Theta(n)$ in the grid to
	$O({\rm polylog}(n))$, where $n$ is the number of vertices in the graph.
Contrary to the improvement in decentralized routing, we show that
	when $\gamma=d$, with high probability
	the small-world graph is not polylogarithmically 
	hyperbolic.
We further show that when $0\le \gamma < d$, the random small-world graphs
	is not hyperbolic and when $\gamma > 3$ and $d=1$, the random graphs
	is not polylogarithmically hyperbolic.
Although there still exists a gap between hyperbolic $\delta$ and graph diameter 
at the sweetspot of $\gamma=d$, our 
results already indicate that long-range edges that enable efficient 
navigation do not significantly improve the hyperbolicity of the graphs.

The second family of graphs is random {\em ringed trees}.
A ringed tree is a binary tree with nodes in each level of the tree connected
	by a ring (Figure~\ref{fig:disk}(d)).
Ringed trees can be viewed as an idealized version of hierarchical structure
	with local peer connections, such as the Internet autonomous system (AS)
	topology.
We show that ringed tree is quasi-isometric to the Poincar\'{e} 
disk, the well known hyperbolic space representation, and thus it
	is constantly hyperbolic.
We then study how random additions of long-range links on the leaves of
	a ringed tree affect the hyperbolicity of random ringed trees.
Note that due to the tree base structure, random ringed trees allow
	efficient routing within $O(\log n)$ steps using tree branches.
Our results show that if the random long-range edges 
	between leaves are added according to a
	probability function that decreases exponentially fast with the ring
	distance between leaves, then the resulting random graph is 
	logarithmically hyperbolic, but if the probability function decreases
	only as a power-law with ring distance, or based on another
	tree distance measure similar to~\cite{Kle01}, the resulting random graph
	is not hyperbolic.
Furthermore, if we use binary trees instead of ringed trees as base graphs,
	none of the above augmentations is hyperbolic.
Taken together, our results indicate that $\delta$-hyperbolicity of graphs
	is quite sensitive to both base graph structures and probabilities of
	long-range connections.

To summarize, 
	we provide one of the first significant analytical 
	results on the hyperbolicity properties of important families of random graphs.
Our results demonstrate that efficient routing performance does not
	necessarily mean good graph hyperbolicity (such as logarithmic hyperbolicity).

\subsection{Related work}
\label{sxn:intro-related}

There has been a lot of work on search and decentralized search subsequent 
to Kleinberg's original work~\cite{Kle00,Kle01}, much of which has been 
summarized in the review~\cite{Kle06}.
In a parallel with this, there has been empirical and theoretical work on 
hyperbolicity of real-world complex networks as well as simple random graph 
models.
On the empirical side,~\cite{Bar02} showed that measurements of the Internet 
are negatively curved;~\cite{JL02,JL04,JLB08,LouTH,LohsoonTH} provided 
empirical evidence that randomized scale-free and Internet graphs are more 
hyperbolic than other types of random graph models;~\cite{NS11} measured 
the average $\delta$ and related curvature to congestion; and~\cite{montgolfier2011treewidth} 
measured treewidth and hyperbolicity properties of the Internet.
On the theoretical side, one has~\cite{PR00,JLB08,BRSV10,NST12,Tucci12,Shang12}, among which
	\cite{NST12,Tucci12,Shang12} study Gromov hyperbolicity of random graphs and
	are most relevant to our work.
In~\cite{NST12}, Narayan et al. study $\delta$-hyperbolicity of 
	sparse Erd\H{o}s-R\'{e}nyi random graphs $G(n,p)$ where $n$ is the number of vertices in
	the graph and $p$ is the probability of any pair of nodes has an edge, with
	$p=c/n$ for some constant $c>1$.
They prove that with positive probability these graphs are not 
	$\delta$-hyperbolic for any positive constant $\delta$ (i.e. not constantly hyperbolic
	in our definition).
In~\cite{Tucci12}, Tucci shows that random $d$-regular graphs are almost surely
	not constantly hyperbolic.
In~\cite{Shang12}, Shang shows that with non-zero probability the 
	Newman-Watts small-world model~\cite{NW99} is not constantly hyperbolic.
These studies only investigate constant hyperbolicity on random graphs, while our study moves
	beyond constant hyperbolicity and show whether certain random graph classes are logarithmically
	hyperbolic, or not hyperbolic at all, comparing with the graph diameters.
Moreover, the one dimensional Newman-Watts small-world model studied in
	\cite{Shang12} is a special case of the Kleinberg small-world model we studied in this
	paper (with dimension $d=1$ and small-world parameter $\gamma=0$).
As given by Theorem~\ref{thm:main}~(2), we show that with probability $1-o(1)$ the hyperbolic
	$\delta$ of these random graphs is $\Omega(\log n)$, where $n$ is the number of vertices
	in the graph.
Therefore, our result is stronger than the result in~\cite{Shang12} for this particular case.

%

More generally, we see two approaches connecting hyperbolicity with efficient routing in
graphs. One approach study efficient computation of graph properties, such as
diameters, centers, approximating trees, and 
packings and coverings for low hyperbolic-$\delta$ graphs and
metric spaces~\cite{chepoi2008diameters,CD00,GL05,CDEHVX10,CE07}.
In large part, the reason for this interest is that there are often direct 
consequences for navigation and routing in these 
graphs~\cite{GL05,CDEHVX10,Kle07,ABKMRT07}.
While these results are 
of interest for general low hyperbolic-$\delta$ graphs, they can be less interesting when applied to 
small-world and other low-diameter random models of complex networks.
To take one example,~\cite{chepoi2008diameters} provides a simple construction of a 
distance approximating tree for $\delta$-hyperbolic graphs on $n$ vertices; 
but the $O(\log n)$ additive-error guarantee is clearly less interesting for 
models in which the diameter of the graph is $O(\log n)$.
Unfortunately, this $O(\log n)$ arises for a very natural reason in the 
analysis, and it is nontrivial to improve it for popular tree-like complex 
network models.

Another approach taken by several recent papers 
is to build random graphs from hyperbolic metric spaces and then shows that such random graphs lead to several 
common properties of small-world complex networks, including good 
navigability properties~\cite{BKC09,PKBV10,KPKVB10,KPVB09}.
While assuming a low hyperbolicity metric space to
	build random graphs in these studies makes intuitive sense, it is 
difficult to prove nontrivial results on the Gromov's $\delta$ 
	of these random graphs even for simple 
	random graph models that are intuitively tree-like.


Understanding the relationship between these two 
approaches was one of the original motivations of 
our research.
In particular, the difficulties in the above two approaches lead us to study hyperbolicity
	of small-world and tree-like random graphs.

Finally, ideas related to hyperbolicity have been applied in numerous other 
networks applications, e.g., to problems such as distance estimation, 
network security, sensor networks, and traffic flow and congestion 
minimization~\cite{ST08,JLBB11,JLHB07,JL04,NS11,BT10a_TR}, as well 
as large-scale data visualization
~\cite{Mun98}.
The latter applications typically take important advantage of the idea that 
data are often hierarchical or tree-like and that there is ``more room'' in 
hyperbolic spaces of dimension 2 than Euclidean spaces of any finite dimension.

\paragraph{Paper organization.}
In Section~\ref{sec:pre} we provide
	basic concepts and terminologies on hyperbolic spaces and graphs
	that are needed in this paper.
In Sections~\ref{sec:smallworld} and~\ref{sec:ringedtree} we study
	the hyperbolicity of small-world random graphs and ringed tree based
	random graphs.
For ease of reading, in each of Sections~\ref{sec:smallworld} and~\ref{sec:ringedtree}
	we first summarize our technical results together
	with their implications (Sections~\ref{sec:swresult} and~\ref{sec:rtresult}), then provide
	the outline of the analyses (Sections~\ref{sec:swoutline} and~\ref{sec:rtoutline}), followed by
	the detailed technical proofs (Sections~\ref{sec:swproof-app} and~\ref{sec:rtproof}), and finally
	discuss extensions of our results
		to other related models (Sections~\ref{sec:sworldextension} and~\ref{sec:rtringextension}).
We discuss open problems and future directions
	related to our study in Section~\ref{sec:discuss}.

\section{Preliminaries on hyperbolic spaces and graphs} \label{sec:pre}

Here, we provide basic concepts concerning hyperbolic spaces and graphs used 
in this paper; for more comprehensive coverage on hyperbolic spaces, see, 
e.g.,~\cite{BH99}.

\subsection{Gromov's $\delta$-hyperbolicity}

In \cite{gromov1987hyperbolic}, Gromov defined a notion of hyperbolic metric 
space; and he then defined hyperbolic groups to be finitely generated groups 
with a Cayley graph that is hyperbolic. 
There are several equivalent definitions (up to a multiplicative constant) 
of Gromov's hyperbolic metric space \cite{bowditch1991notes}. 
In this paper, we will mainly use the following. 

\begin{defi}[Gromov's four-point condition]
In a metric space $(X,d)$, given $u,v,w,x$ with $d(u,v)+d(w,x) \geq 
	d(u,x)+d(w,v) \geq d(u,w)+d(v,x)$ in $X$, 
	we note $\delta(u,v,w,x)=(d(u,v)+d(w,x)-d(u,x)-d(w,v))/2$. 
$(X,d)$ is called \emph{$\delta$-hyperbolic} for some non-negative 
	real number $\delta$ if for any four points 
	$u,v,w,x \in X$, $\delta(u,v,w,x) \leq \delta$. 
Let $\delta(X,d)$ be the smallest possible value of such $\delta$, which
	can also be defined as $\delta(X,d) = \sup_{u,v,w,x \in X} \delta(u,v,w,x)$.
\end{defi}

Given an undirected, unweighted and connected graph $G=(V,E)$, one can view 
it as a metric space $(V, d_G)$, where $d_G(u,v)$ denotes the (geodesic) 
graph distance between two vertices $u$ and $v$.
Then, one can apply the above four point condition to define its 
$\delta$-hyperbolicity, which we denote $\delta = \delta(G) = \delta(V,d_G)$ 
(and which we sometimes refer to simply as the hyperbolicity or the $\delta$ 
of the graph).
Trees are $0$-hyperbolic; and $0$-hyperbolic graphs are exactly clique trees 
(or called block graphs), which can be viewed as cliques connected in a tree 
fashion~\cite{Ho79}.
Thus, it is often helpful to view graphs with a low hyperbolic $\delta$ as 
``thickened'' trees, or in other words, as tree-like when viewed at large 
size scales.

If we let $D(G)$ denote the diameter of the graph $G$, then, by the triangle 
inequality, we have $\delta(G) \le D(G)/2$.
We will use the asymptotic difference between the hyperbolicity $\delta(G)$ 
and the diameter $D(G)$ to characterize the hyperbolicity of the graph $G$.
\begin{defi}[Hyperbolicity of a graph]
For a family of graphs $\cal G$ with diameter $D(G), G\in {\cal G}$ going to infinity as the size of $G$ grows to infinity, we say 
that graph family $\cal G$ is
{\em constantly (resp. logarithmically, polylogarithmically, or weakly) hyperbolic}, 
	if $\delta(G) = O(1)$ (resp. $O(\log D(G))$, $O((\log D(G))^c)$ for some constant $c>0$, 
	or $o(D(G))$) when $D(G)$ goes to infinity; and 
	$\cal G$ is {\em not hyperbolic} if $\delta(G) = \Theta(D(G))$, where $G\in {\cal G}$.
\end{defi}


\noindent
The above definition provides more fine-grained characterization of hyperbolicity of graph families
	than one typically sees in the literature, which only discusses whether or not a graph family
	is constantly hyperbolic.
This definition does not address the hyperbolicity of graph families where the diameter 
	stays bounded while the size of the graph goes unbounded.
For these graph families, one may probably need tight analysis on the constant factor
	between the hyperbolic $\delta$ and the graph diameter, and it is out of the scope of this paper.
	

\subsection{Rips condition}

Rips condition~\cite{gromov1987hyperbolic,BH99} is a technically equivalent
	condition to the Gromov's four point condition up to a constant factor.
We use the Rips condition when analyzing
	the $\delta$-hyperbolicity of ringed trees.
In a metric space $(X,d)$, we define a \emph{geodesic segment} $[u,v]$ between two points $u,v$ to be the image of a function $\rho:[0,d(u,v)] \to [u,v]$ satisfying $\rho(0)=u$, $\rho(d(u,v))=v$, $d(\rho(s),\rho(t))=|s-t|$ for any $s,t \in [0,d(u,v)]$. We say that a metric space is \emph{geodesic} if every pair of its points has a geodesic segment, not necessarily unique.
In a geodesic metric space $(X,d)$, given $u,v,w$ in $X$, we denote $\Delta(u,v,w)=[u,v] \cup [v,w] \cup [w,u]$ a \emph{geodesic triangle}. $[u,v]$, $[v,w]$, $[w,u]$ are called \emph{sides} of $\Delta(u,v,w)$. We should note that, in general, geodesic segments and geodesic triangles are not unique up to their endpoints.

In a metric space, it is sometimes convenient to consider distances between point sets in the following way. We say that a set $S$ is within distance $d$ to another set $T$ if $S$ is contained in the ball $B(T,d)$ of all points within distance $d$ to some point in $T$. We say that $S$ and $T$ are within distance $d$ to each other if $S$ is within distance $d$ to $T$ and vice versa.


\begin{defi}[Rips condition]
A geodesic triangle $\Delta(u,v,w)$ in a geodesic metric space $(X,d)$
	is called \emph{$\delta$-slim} for some non-negative real number $\delta$
	if any point on a side is within 
	distance $\delta$ to the union of the other two sides. 
$(X,d)$ is called \emph{Rips $\delta$-hyperbolic} if every geodesic triangle 
	in $(X,d)$ is $\delta$-slim. 
We denote $\delta_{Rips}(X,d)$ the smallest possible value of such $\delta$
	(could be infinity).
\end{defi}

\noindent
It is known (see, e.g., \cite{ghys1990groupes,BH99,chepoi2008diameters}) that
	$\delta(X,d)$ and $\delta_{Rips}(X,d)$ differ only within a multiplicative 
	constant. 
In particular, $\delta(X,d) \le 8\delta_{Rips}(X,d)$ and
	$\delta_{Rips}(X,d) \le 4\delta(X,d)$.
Since we are only concerned with asymptotic growth of $\delta(X,d)$, 
	Rips condition can be used in place of the Gromov's 
	four point condition.

For an undirected unweighted graph $G=(V,E)$, 
	we can also treat it as a geodesic metric space with every edge
	interpreted as a segment of length $1$, and thus use
	the Rips condition to define its hyperbolicity, which we denote
	as $\delta_{Rips}(G)$. 
Note that in the case of unweighted graph, when considering the distance
	between two geodesics on the graph, we only consider the distance
	among the vertices, since other points on the edges can add at most
	$2$ to the distance between vertices.


\subsection{Poincar\'{e} disk}

\begin{figure}
   \begin{center}
	\begin{tabular}{cccc}
         \includegraphics[height=1.2in]{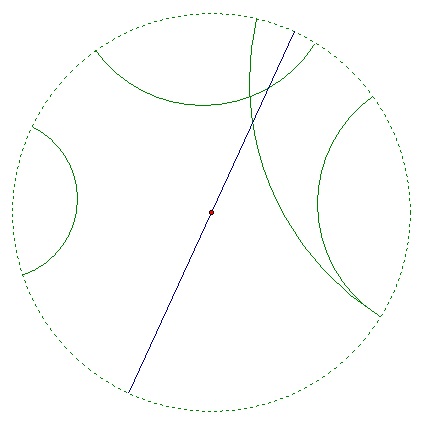} &
         \includegraphics[height=1.2in]{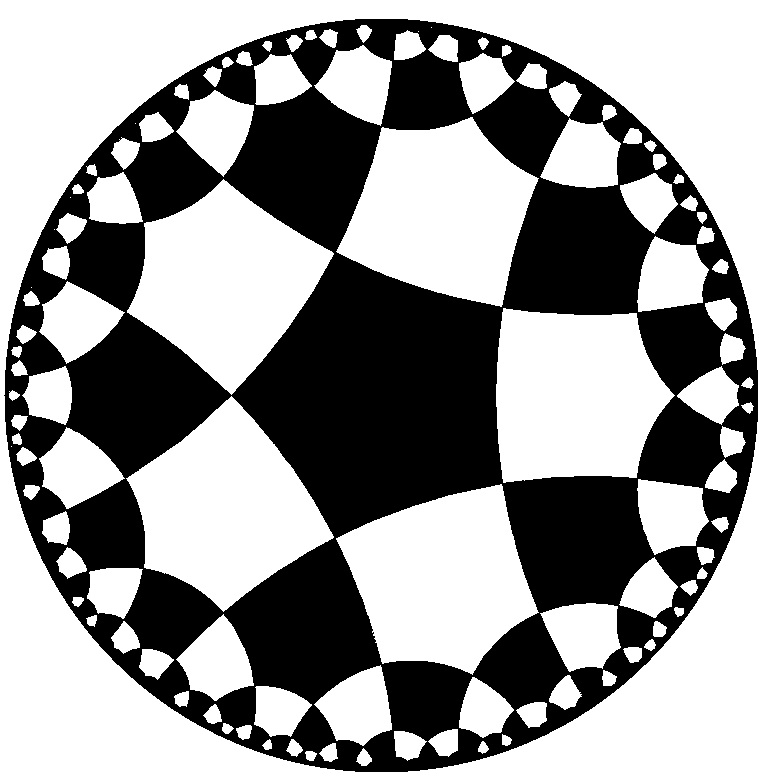} &
         \includegraphics[height=1.2in]{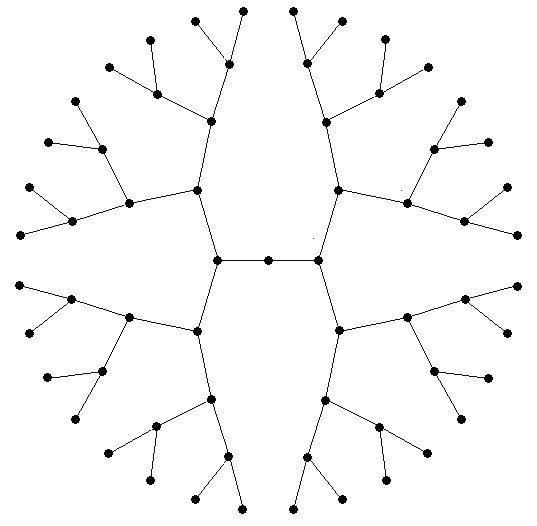} &
         \includegraphics[height=1.2in]{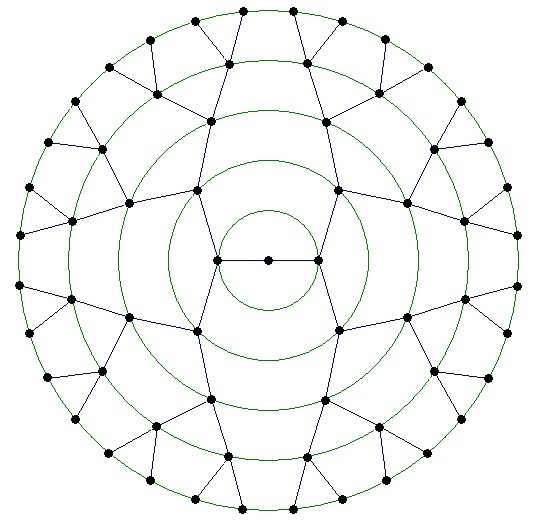} \\
	(a) Poincar\'{e} disk &
	(b) \parbox[t]{0.8in}{Tessellation of Poincar\'{e} disk} &
	(c) Binary tree &
	(d) Ringed tree 
	\end{tabular}
   \end{center}
\caption{ Poincar\'{e} disk, its tessellation, a binary tree, 
	and a ringed tree.
}
\label{fig:disk}
\end{figure}

The Poincar\'{e} disk (see Figure~\ref{fig:disk}(a) for an illustration) is 
a well-studied hyperbolic metric space.
Although in this paper we touch upon it only briefly when we study 
ringed-tree graphs, it is useful to convey intuition about hyperbolicity and
tree-like behavior.

\begin{defi}
Let $D=B(0,1)$ be a open disk on the complex plane with origin $0$ and
	radius $1$, with the following distance function:
\[ d(u,v) = \mathrm{arccosh}\bigg(1+\frac{2\Vert u - v \Vert^{2}}{(1-\Vert u \Vert^{2})(1-\Vert v \Vert^{2})} \bigg). \]
$(D,d)$ is a metric space. We call it the Poincar\'{e} disk.
\end{defi}


\noindent
Visually, a (hyperbolic) line in the Poincar\'{e} disk is the segment of a 
circle in the disk that is perpendicular to the circular boundary of the 
disk, and thus all lines bend inward towards the origin.
The hyperbolic distance between two points in the disk of fixed distance
	in the complex plane increase exponentially fast when they moves
	towards the boundary of the disk, meaning that there is much more
	``space'' towards the boundary than around the origin.
This can be seen from a tessellation of the Poincar\'{e} disk, as shown
	in Figure~\ref{fig:disk}(b).

\subsection{Quasi-isometry}

Quasi-isometry, defined as follows, is a concept used to capture the 
large-scale similarity between two metric spaces.
\begin{defi}[Quasi-isometry]
For two metric spaces $(X,d_{X}), (Y,d_{Y})$, we say that $f: X \to Y$ is a $(\lambda,\epsilon)$-quasi-isometric embedding from $X$ to $Y$ if for any 
	$u,v \in X$,
\[ \frac{1}{\lambda}d_{X}(u,v) - \epsilon \leq d_{Y}(f(u),f(v)) \leq \lambda d_{X}(u,v) + \epsilon. \]
Furthermore, if the $\epsilon$ neighborhood of $f(Y)$ covers $X$, then we 
say that $f$ is a $(\lambda,\epsilon)$-quasi-isometry.
Moreover, we say that $X,Y$ are quasi-isometric if such a 
$(\lambda,\epsilon)$-quasi-isometry exists for some constants $\lambda$ and~$\epsilon$.
\end{defi}

\noindent
If two metric spaces are quasi-isometric with some constant, then they have 
the same ``large-scale'' behavior. 
For example, the $d$-dimensional grid $\mathbb{Z}^d$ and the $d$-dimensional 
Euclidean space $\mathbb{R}^d$ are quasi-isometric, realized by the 
$(\sqrt{d}, \sqrt{d}/2)$-quasi-isometric embedding $(x,y) \mapsto (x,y)$.
As a second example, consider an infinite ringed-tree:
start with a binary tree (illustrated in Figure~\ref{fig:disk}(c)) and then 
connect all vertices at a given tree level into a ring.
This is defined more formally in Section~\ref{sec:ringedtree}, but an example
is illustrated in Figure~\ref{fig:disk}(d).
As we prove in Section~\ref{sec:ringedtree}, the infinite ringed tree is quasi-isometric to the 
Poincar\'{e} disk---thus it may be equivalently viewed as a ``softened'' 
binary tree or as a ``coarsened'' Poincar\'{e} disk.

Quasi-isometric embeddings have the important property of preserving 
hyperbolicity, up to a constant factor, as given by the following proposition.

\begin{prop}[Theorem 1.9, Chapter III.H of \cite{BH99}]\label{prop:quasi}
Let $X$ and $X'$ be two metric spaces and let $f: X' \rightarrow X$ be a 
	$(\lambda, \epsilon)$-quasi-isometric embedding.
If $X$ is $\delta$-hyperbolic, then $X'$ is $\delta'$-hyperbolic, where
	$\delta'$ is a function of $\delta$, $\lambda$, and $\epsilon$.
\end{prop}

\section{$\delta$-hyperbolicity of grid-based small-world graphs}\label{sec:smallworld}


In this section, we consider the $\delta$-hyperbolicity of graphs constructed 
according to the small-world graph model as formulated by
Kleinberg~\cite{Kle00}, in which long-range edges are added on top of a base 
grid, which is a discretization of a low-dimensional Euclidean space.

The model starts with $n$ vertices forming a $d$-dimensional base grid 
(with wrap-around).  
More precisely, given positive integers $n$ and $d$ such that $n^{1/d}$ is also an integer, 
	let $B=(V,E)$ be the base grid, with
	$V=\{(x_1,x_2,\ldots,x_d) $ $\ |\ $ 
	$x_i \in \{0,1,\ldots, n^{1/d}-1\}, i\in [d]\}$,
        $E = \{((x_1,x_2,\ldots,x_d),$ $ (y_1,y_2,\ldots,y_d))\ |\ $ 
	$\exists j\in [d], y_j = x_j+1 \mod n^{1/d} $ ${\rm\ or\ } y_j = x_j-1 \mod n^{1/d},
	\forall i\ne j, y_i=x_i\}$.
Let $d_B$ denote the graph distance metric on the base grid $B$.
We then build a random graph $G$ on top of $B$, such that
	$G$ contains all vertices and all edges (referred to as grid edges)
	of $B$, and for each node $u\in V$, 
	it has one long-range edge (undirected)
	connected to some node $v\in V$, with probability proportional
	to $1/d_B(u,v)^\gamma$, where $\gamma\ge 0$ is a parameter.
We refer to the probability space of these random graphs as
	$\KSW(n, d, \gamma)$; and
we let $\delta(\KSW(n, d, \gamma))$ denote the random variable of
	the hyperbolic $\delta$ of a randomly picked graph $G$ in 
	$\KSW(n, d, \gamma)$.
Recall that Kleinberg showed that the small-world graphs with $\gamma=d$ 
allow efficient decentralized routing (with $O(\log^2 n)$ routing hops in 
expectation), whereas graphs with $\gamma \ne d$ do not allow any efficient 
decentralized routing (with $\Omega(n^c)$ routing hops for some constant 
$c$)~\cite{Kle00}; and 
note that the base grid $B$ has large hyperbolic $\delta$, i.e., 
$\delta(B) = \Theta(n^{1/d}) = \Theta(D(B))$.
Intuitively, the structural reason for the efficient routing performance
	at $\gamma=d$ is that long-range edges are added ``hierarchically'' 	
	such that each node's 
	long-range edges are nearly uniformly distributed over all ``distance 
	scales''.



\subsection{Results and their implications} \label{sec:swresult}

The following theorem summarizes our main technical results on the 
hyperbolicity of small-world graphs for different combinations of $d$ and 
$\gamma$.

\begin{thm} \label{thm:main}
With probability $1- o(1)$ (when $n$ goes to infinity), we have
\begin{enumerate}
\item
$\delta(\KSW(n, d, \gamma)) = \Omega((\log n)^{\frac{1}{1.5(d+1)+\varepsilon}})$
	when $d\ge 1$ and $\gamma = d$, for any $\varepsilon>0$
	independent of~$n$;
\item
$\delta(\KSW(n, d, \gamma)) = \Omega(\log n)$ when   
	$d\ge 1$ and $0\le \gamma < d$; and
\item
$\delta(\KSW(n, d, \gamma)) = \Omega(n^{\frac{\gamma-2}{\gamma-1}-\epsilon})$ 
	when $d=1$ and
	$\gamma > 3$, 
        for any 
        $\epsilon > 0$ independent of~$n$.
\end{enumerate}
\end{thm}

This theorem, together with the results of~\cite{Kle00} on the navigability of
small-world graphs, have several implications.
The first result shows that when $\gamma = d$, 
	with high probability the hyperbolic
	$\delta$ of the small-world graphs is at least
	$c(\log n)^{\frac{1}{1.5(d+1)}}$ for some constant $c$.
We know that the diameter is $\Theta(\log n)$ in expectation
	when $\gamma=d$~\cite{MN04}.
Thus the small-world graphs at the sweetspot for efficient routing is not 
polylogarithmically hyperbolic, i.e., $\delta$ is not $O(\log^c\log n)$-hyperbolic for
any constant $c>0$.
However, there is still a gap between our lower bound the upper bound 
provided by the diameter, and thus it is still open whether small-world 
graphs are weakly hyperbolic or not hyperbolic.
Overall, though, our result indicates no drastic improvement on the 
hyperbolicity (relative to the improvement of the diameter) for small-world 
graphs at the sweetspot (where a dramatic improvement was obtained for the
efficiency of decentralized routing).

The second result shows that when $\gamma < d$, then $\delta=\Omega(\log n)$.
The diameter of the graph in this case is $\Theta(\log n)$~\cite{MN04}; 
thus, we see that when $\gamma < d$ the hyperbolic $\delta$ is asymptotically 
the same as the diameter, i.e., although $\delta$ decreases as edges are 
added, small-world graphs in this range are not hyperbolic.
The third result concerns the case $\gamma > d$, in which case the
random graph degenerates towards the base grid (in the sense that most of all 
of the long-range edges are very local), which itself is not hyperbolic.
For the general $\gamma$, we show that for the case of
	 $d=1$ the hyperbolic $\delta$ is lower bounded by a (low-degree)
	polynomial of $n$; this also 
implies that the graphs in this range are not polylogarithmically hyperbolic.
Note that our polynomial exponent $\frac{\gamma-2}{\gamma-1}-\epsilon$
	matches the diameter lower bound proven in~\cite{NM05}.

\subsection{Outline of the proof of Theorem~\ref{thm:main}} 
\label{sec:swoutline}

In this subsection, we provide a summary of the proof of 
Theorem~\ref{thm:main}.
In our analysis, we use two different techniques, one for the first two 
results in Theorem~\ref{thm:main}, and the other for the last result; 
in addition, for the first two results, we further divide the analysis into 
the two cases $d\ge 2$ and $d=1$.

When $d\ge 2$ and $0 \le \gamma \le d$, the main idea of the proof is to 
pick a square grid of size $\ell_0$ (it does not matter in which dimension 
the square is picked from). 
We know that when only grid distance is considered, 
	the four corners of the square grid
	have the Gromov $\delta$ value equal to $\ell_0$.
We will show that, as long as $\ell_0$ is not very large 
	(to be exact, $O((\log n)^{\frac{1}{1.5(d+1)+\varepsilon}})$ when $\gamma = d$ and
	$O(\log n)$ when $0 \le \gamma < d$), 
	the probability that any pair of vertices on this square grid have
	a shortest path shorter than their grid distance after
	adding long-range edges is close to zero (as $n$ tends to
	infinity).
Therefore, with high probability, the four corners selected have
	Gromov $\delta$ as desired in the lower bound results.

To prove this result, we study the probability that any pair of
	vertices $u$ and $v$ at grid distance $\ell$ 
	are connected with a path that contains at least one long-range
	edge and has length 
	at most $\ell$.
We upper bound such $\ell$'s 
	so that this probability is close to zero.
To do so, we first classify such paths into a number of categories, based 
	on the pattern of paths connecting $u$ and $v$: how it alternates
	between grid edges and long-range edges, and the direction on each
	dimension of the grid edges and long-range edges (i.e., whether it is
	the same direction as from $u$ to $v$ in this dimension, or the opposite
	direction, or no move in this dimension).
We then bound the probability of existing a path in each category and finally 
	bound all such paths in aggregate.
The most difficult part of the analysis is the bounding of the probability
	of existing a path in each category.

For the case of $d = 1$ and $0 \le \gamma \le d$, the general idea is similar
	to the above.
The difference is that we do not have a base square to start with.
Instead, we find a base ring of length $\Theta(\ell_0)$ using one
	long-range edges $e_0$, where $\ell_0$ is fixed to be the same
	as the case of $d\ge 2$.
We show that with high probability, 
	(a) such an edge $e_0$ exists, and
	(b) the distance of any two vertices on the ring is simply their ring 
	distance.
This is enough to show the lower bound on the hyperbolic $\delta$.


For the case of $\gamma > 3$ and $d=1$, a different technique is used to
	prove the lower bound on hyperbolic $\delta$.
We first show that, in this case, with high probability all long-range
	edges only connect two vertices with ring distance at most some
	$\ell_0 = o(\sqrt{n})$.
Next, on the one dimensional ring, we first find two vertices $A$ and
	$B$ at the two opposite ends on the ring.
Then we argue that there must be a path ${\cal P}_{AB}^+$ that only 
	goes through the clockwise side of ring from $A$ to $B$, 
	while another path ${\cal P}_{AB}^-$ that
	only goes through the counter-clockwise side of the ring from $A$ to $B$,
	and importantly, the shorter length of 
	these two paths are at most $O(\ell_0)$
	longer than the distance between $A$ and $B$.
We then pick the middle point $C$ and $D$ of ${\cal P}_{AB}^+$ and
	${\cal P}_{AB}^-$, respectively, and argue that the $\delta$ value
	of the four points
	$A$, $B$, $C$, and $D$ give the desired lower~bound.

\subsection{Detailed proof of Theorem~\ref{thm:main}} 
\label{sec:swproof-app}

\subsubsection{The case of $d\ge 2$ and $0\le \gamma \le d$}\label{sec:d2}

For this case, let $n'=n^{1/d}$ be the number of vertices on one side of the grid.
For convenience, our main analysis for this case uses $n'$ instead of $n$.

\vspace{\topsep}

\noindent
{\bf Lemmas for calculation.}\ \ 
We first provide a couple of lemmas used in our probability 
	calculation.

\begin{lem} \label{lemma:calc1}
There exists a constant $c_1$, such that for any $k,m\in\mathbb{Z}^+$, we
	have
\[\sum_{y_1+\cdots+y_k=m \atop y_1,\ldots,y_k\in\mathbb{Z}^+}\frac{1}{y_1y_2\cdots y_k}\leq\frac{(c_1\ln m)^{k-1}}{m},\]
where the left side is considered to be 0 for $k>m$; and for $k=m=1$, 
	the right side $0^0$ is considered to be 1.
\end{lem}

\begin{proof}
For $k=1$, it is trivial. For $k=2$, we have
\begin{align*}
\sum_{y_1+y_2=m \atop y_1,y_2\in\mathbb{Z}^+}\frac{1}{y_1y_2}&\leq2\left(\frac{1}{\lfloor m/2\rfloor\cdot\lceil m/2\rceil}+\cdots+\frac{1}{(m-1)\cdot1}\right)\leq\frac{2}{\lfloor m/2\rfloor}\left(\frac{1}{\lceil m/2\rceil}+\cdots+\frac{1}{1}\right) \\
&<c_1\frac{\ln m}{m},
\end{align*}
where $c_1$ is roughly 4.

Suppose the lemma holds for $k-1$, with $k\ge 3$.
The induction hypothesis is
\[\sum_{y_1+\cdots+y_{k-1}=m \atop y_1,\ldots,y_{k-1}\in\mathbb{Z}^+}\frac{1}{y_1y_2\cdots y_{k-1}}\leq\frac{(c_1\ln m)^{k-2}}{m},\] Since the logarithm function is increasing, we have 
\begin{align*}
\sum_{y_1+\cdots+y_k=m \atop y_1,\ldots,y_k\in\mathbb{Z}^+}\frac{1}{y_1y_2\cdots y_k} & \leq \frac{1}{1}\frac{(c_1\ln(m-1))^{k-2}}{m-1}+\cdots+\frac{1}{m-1}\frac{(c_1\ln 1)^{k-2}}{1} \\
& \leq (c_1\ln m)^{k-2}\cdot\sum_{y_1+y_2=m \atop y_1,y_2\in\mathbb{Z}^+}\frac{1}{y_1y_2} \\
& \leq (c_1\ln m)^{k-2}\cdot \frac{c_1\ln m}{m}.
\end{align*}
Therefore the inequality holds for all $k$.
\end{proof}

\begin{lem} \label{lemma:calc2}
For any constant $\theta\in \mathbb{R}$ with $0\le \theta < 1$,
	there exists a constant $c_2$ (may only depend on $\theta$), such that
	for any constants $k,n'\in\mathbb{Z}^+$, $m\in\mathbb{R}$, and
	non-zero $\lambda_1,\lambda_2,\ldots,\lambda_k\in\mathbb{R}$, we have
\[\sum_{\lambda_1y_1+\cdots+\lambda_ky_k=m \atop y_1,\ldots,y_k\in\{1,2,\ldots,n'\}}\frac{1}{y_1^\theta y_2^\theta \cdots y_k^\theta}\leq(c_2n')^{(k-1)(1-\theta)},\]
where the left side is considered to be 0 if there is no 
	$y_1,y_2,\ldots,y_k\in\{1,2,\ldots,n'\}$ satisfying 
	$\lambda_1y_1+\lambda_2y_2+\cdots+\lambda_ky_k=m$.
\end{lem}

\begin{proof}
For each tuple $(y_1,y_2,\ldots,y_{k-1})\in\{1,2,\ldots,n'\}^{k-1}$, there is at most one $y_k\in\{1,2,\ldots,n'\}$ satisfying $\lambda_1y_1+\lambda_2y_2+\cdots+\lambda_ky_k=m$. 
Since $\frac{1}{y_k^\theta}\leq1$ because $0\le \theta < 1$, we have 
\begin{align*}
\sum_{\lambda_1y_1+\cdots+\lambda_ky_k=m \atop y_1,\ldots,y_k\in\{1,2,\ldots,n'\}}\frac{1}{y_1^\theta y_2^\theta \cdots y_k^\theta}&\leq\sum_{y_1,\ldots,y_{k-1}\in\{1,2,\ldots,n'\}}\frac{1}{y_1^\theta y_2^\theta \cdots y_{k-1}^\theta} \\
&=\left(\sum_{i=1}^{n'}\frac{1}{i^\theta}\right)^{k-1}\leq(c_2n')^{(k-1)(1-\theta)},
\end{align*}
where $c_2$ is roughly $(\frac{1}{1-\theta})^{\frac{1}{1-\theta}}$. Therefore the lemma is proved.
\end{proof}

\vspace{\topsep}
\noindent
{\bf Classification of paths.}\ \ 
In a $d$-dimensional random graph $\KSW(n,d,\gamma)$, 
	there are two kinds of edges: {\em grid edges}, which are 
	edges on the grid, and {\em long-range edges}, which are randomly added. 


Fix two vertices $u$ and $v$, a path from $u$ to $v$ may contain some 
	long-range edges and some grid edges. 
We divide the path into several {\em segments} along the way from $u$ to $v$:
	(a) each segment is either one long-range edge (called a {\em long-range
	segment}) or a batch of consecutive 
	grid edges (called a {\em grid segment}); and
	(b) two consecutive segments  cannot be both grid segments
	(otherwise combining them into one segment).

We use a $d$-dimensional vector to denote each edge, so that the source coordinate plus this vector equals to the destination coordinate module $n'$. For grid with wrap-around, there may be multiple vectors corresponding to one edge. We choose the vector in which every element is from 
	$\{-\lfloor\frac{n'}{2}\rfloor,-\lfloor\frac{n'}{2}\rfloor+1,
	\ldots,\lfloor\frac{n'-1}{2}\rfloor\}$.
In this way, the vector representation of each edge is unique and the absolute value of every dimension is the smallest. We call this the {\em edge vector} of that edge. For every segment in the path, we call the summation (not module $n'$) of all edge vectors the {\em segment vector}. 
For a vector $(x_1,x_2,\ldots,x_d)$, define its {\em sign pattern} as $(\sgn(x_1), \sgn(x_2), \ldots, \sgn(x_d))$.

We say two paths from $u$ to $v$ belong to the same {\em category} if
	(a) they have the same number of segments; 
	(b) their corresponding segments are of the same type (long-range
	or grid segments); 
	(c) for every pair of corresponding long-range segments in the two paths, 
	the sign patterns of their segment vectors are the same;
	(d) for every pair of corresponding grid segments in the two paths,
	their segment vectors are equal; and
	(e) the summations (not module $n'$) of all segment vectors in the 
	two paths are equal.

The last condition is only used to distinguish paths that go different rounds in each dimension on grid with wrap-around.

In one category, there exist paths of which the long-range edges are
	identical but the grid edges may be different.
To compute the probability of existing a path in a category, we only
	need to consider one path among the paths with identical long-range
	edges, since grid edges do not change probabilistic events and thus
	one such path exists if any only if other such paths exist.
	
We also assume that there are no repeated long-range edges in every path. For a path that has repeated long-range edges, we can obtain a shorter subpath without any repeated long-range edges so that the original path exists if and only if the new one exists. Since we are going to calculate the probability about paths not exceeding some length, it is safe to only consider paths without repeated long-range edges.



\begin{lem} \label{lemma:numcate2}
There exists a constant $c_3$ (dependent on $d$) such that
	for any fixed $\ell$, the number of categories of paths from
	$u$ to $v$ of
	length $\ell$ is at most ${c_3}^\ell$.
\end{lem}
\begin{proof}
For each edge on the path, if it is a grid edge, it could be in one of
	the $d$ dimensions, and in each dimension it could be in one of the
	two opposite directions, and thus a grid edge has $2d$ possibilities.
If the edge is a long-range edge, on each dimension its sign has three
	possibilities $(+1,0,-1)$, so totally $3^d$ possibilities for
	the sign pattern of the long-range segment vector.
Moreover, in the grid with wrap-around, we consider each wrap-around 
	of the path on some dimension to be one round in that dimension.
Then the path can go at most $2\ell+1$ different numbers of rounds on each dimension (ranging from $\ell$ rounds in one direction up to $\ell$ rounds in the other direction), so the summation of all segment vectors has at most $(2\ell+1)^d$ different values.
The choice of each edge out of $2d+3^d$ possibilities and the total summation of segments vectors determine a category.
Therefore, the number of categories is bounded by $(2d+3^d)^\ell(2\ell+1)^d<{c_3}^\ell$ for some $c_3$.
\end{proof}

The above bound on the number of categories is not tight enough to be
	used for later analysis, when the number of long-range segments
	are small.
Thus, we further bound the number of categories in the following~way.

\begin{lem} \label{lemma:numcate1}
There exists a constant $c_4$ (dependent on $d$), such that for
	any fixed $\ell<n'$ and $k$ with $1\leq k\leq \ell$, 
	the number of categories of paths from $u$ to $v$ of length 
	$\ell$ and having $k$ long-range segments is 
	at most ${c_4}^k\ell^{(k+1)(d+1)}/k^{kd}$.
\end{lem}

\begin{proof}
For a path from $u$ to $v$ of length $\ell$ and containing $k$ long-range 
	segments ($1\leq k\leq \ell<n'$), the summation of all segment vectors has at most $(2k+1)^d$ choices. This is because the path can go at most $2k+1$ different number of rounds on each dimension ($k$ rounds in one direction to $k$ rounds in the other direction). We consider the number of categories for a fixed summation of segment vectors first.
	
	Suppose there are $t$ 
	grid segments, each having $a_1,a_2,\ldots,a_t$ edges respectively 
	($t\leq k+1$, $a_i\ge 1$, $a_1+a_2+\cdots+a_t<\ell$). 
If $t=0$, then $k=\ell$, and it is easy to see that there are at most
	$(3^d)^k $ categories.
Suppose now $t \ge 1$.
For the $i$-th grid segment with $a_i$ grid edges, its segment vector
	is such that on each dimension the only possible values 
	are $-a_i, -a_i+1, \ldots, 0, \ldots, a_i-1, a_i$.
Thus, the number of possible segment vectors is $(2a_i+1)^d \le 3^da_i^d$.
Since each long-range edge has $3^d$ possible sign patterns, 
	the number of categories for fixed $t$ and 
	$a_1,a_2,\ldots,a_t$ is at most
	$(3^d)^k\prod_{i=1}^t3^da_i^d<9^{(k+1)d}(\ell/t)^{td}$, where the
	inequality of arithmetic and geometry means is used.

The tuple $(a_1,a_2,\ldots,a_t)$ has less than $\ell^t$ possibilities. 
Considering $(2k+1)^d$ different possibilities of segment vectors summations, the total number of categories from $u$ to $v$ with length $\ell$ 
	and containing $k$ long-range edges is at most
\begin{eqnarray*}
&& (2k+1)^d\left\{(3^d)^k+\sum_{t=1}^{k+1}\ell^t\cdot9^{(k+1)d}(\ell/t)^{td}\right\} \\
& < & (2k+1)^d\left\{3^{dk}+(k+1)9^{(k+1)d}\max_{1\leq t\leq k+1}\{\ell^t(\ell/t)^{td}\}\right\} \\
& = & (2k+1)^d\left\{3^{dk}+(k+1)9^{(k+1)d}\ell^{k+1}(\frac{\ell}{k+1})^{(k+1)d}\right\} \\
& < & c_4^k\cdot\frac{\ell^{(k+1)(d+1)}}{k^{kd}},
\end{eqnarray*}
where $c_4$ is a constant depending only on $d$. 
\end{proof}

\vspace{\topsep}
\noindent
{\bf Probability calculation.}\ \ 
We first give a lemma to calculate the probability of the existence of a specific edge. For an integer $x$, we define $\overline{x}$ to be 
	$|x|$ if $x\neq 0$ and $1$ if $x=0$.
We also define function $f(n')$ as follows:
\[
f(n') = \begin{cases}
\ln n' & \gamma=d, \\
(n')^{d-\gamma} & 0\le \gamma < d.
\end{cases}
\]

\begin{lem} \label{lemma:probedge}
For two vertices $u$ and $v$, the probability of the existence of a
	long-range undirected edge between $u$ and $v$ is at most 
	$c_5(\overline{x_1}\cdot\overline{x_2}\cdots\overline{x_d})^{-\frac{\gamma}{d}}
	/ f(n')$,
	 where $c_5$ is a constant depending only on $d$ and $\gamma$, and $(x_1,x_2,\ldots,x_d)$ is the edge vector if there exists a long-range edge from $u$ to $v$ and depends only on $u$ and~$v$.
\end{lem}

\begin{proof}
Say the non-zero elements of $(x_1,x_2,\ldots,x_d)$ are $(x_{i_1},x_{i_2},\ldots,x_{i_{d'}})$ ($d'\leq d$). Let $p$ be the probability to add an edge from $u$ to $v$. Then
\begin{align*}
p&=\frac{(|x_1|+|x_2|+\cdots+|x_d|)^{-\gamma}}{\Theta(\sum_{i=1}^{n'}\frac{i^{d-1}}{i^\gamma})}=O\left(\frac{|x_{i_1}\cdot x_{i_2}\cdots x_{i_{d'}}|^{-\frac{\gamma}{d'}}}{\sum_{i=1}^{n'} i^{d-1-\gamma}}\right) \\
&\leq O\left(\frac{|x_{i_1}\cdot x_{i_2}\cdots x_{i_{d'}}|^{-\frac{\gamma}{d}}}{\sum_{i=1}^{n'} i^{d-1-\gamma}}\right)=O\left(\frac{(\overline{x_1}\cdot\overline{x_2}\cdots\overline{x_d})^{-\frac{\gamma}{d}}}{\sum_{i=1}^{n'} i^{d-1-\gamma}}\right)=O\left(\frac{(\overline{x_1}\cdot\overline{x_2}\cdots\overline{x_d})^{-\frac{\gamma}{d}}}{f(n')}\right).
\end{align*}

Moreover, the edge may also be from $v$ to $u$, which also has probability $p$. By union bound the probability of the undirected 
	edge $(u,v)$ is $O(p)$.
\end{proof}

The following lemma gives the probability that one edge jumps within a local area.

\begin{lem} \label{lemma:localjump}
For a vertex $u$ and a long-range edge $(u,v)$ from $u$, the probability that the grid distance between $u,v$ is less than $s$ is at most $c_6f(s)/f(n')$, where $c_6$ is a constant depending only on $d$ and $\gamma$.
\end{lem}

\begin{proof}
By union bound, the probability is at most
\[O\left(\sum_{i=1}^si^{d-1}\frac{i^{-\gamma}}{f(n')}\right)=O\left(\sum_{i=1}^s\frac{i^{d-1-\gamma}}{f(n')}\right)\leq c_6\frac{f(s)}{f(n')},\]
where $c_6$ is a constant.
\end{proof}

Given a path category $\cal C$, we now calculate the probability of 
	existing a path in $\cal C$.

\begin{lem}\label{lemma:cateogoryexists}
Given a path category $\mathcal{C}$ with length $\ell$ and $k$ 
	long-range edges.
The probability that there exists a path in $\cal C$ is at most
\[
\left\{\begin{aligned}
	&c_5^k(c_7^kk^k)^d\left/(\ln n')^k\right. && \gamma =d, \\
	&c_5^kc_2^{(k-1)(d-\gamma)}\left/(n')^{d-\gamma}\right. && 0 \le \gamma < d,
\end{aligned}\right.
\]
where $c_2$ and $c_5$ are the constants given in Lemma~\ref{lemma:probedge} 
	and $c_7$ is another constant.
\end{lem}
\begin{proof}
Let the segment vectors of the long-range edges in a path $P\in\mathcal{C}$ be 
	\[(x_{11},x_{12},\ldots,x_{1d}), 
	(x_{21},x_{22},\ldots,x_{2d}), \ldots,(x_{k1},x_{k2},\ldots,x_{kd}).\]
By our definition of a category, all paths in $\cal C$ have the same
	sign patterns on the corresponding long-range segment vectors.
Thus, we can define the following sets for the category $\cal C$:
	$A_i^+=\{j\mid x_{ji}>0,1\leq j\leq k\}$, 
	$A_i^-=\{j\mid x_{ji}<0,1\leq j\leq k\}$, and 
	$A_i^0=\{j\mid x_{ji}=0,1\leq j\leq k\}$, for all
	$1\le i \le d$. 
For fixed $u$ and $v$, there is a fixed vector $(t_1,t_2,\ldots,t_d)$ 
	such that the summation of the segment vectors of 
	all long-range segments in any $P\in\mathcal{C}$ is 
	vector $(t_1,t_2,\ldots,t_d)$. 
This is because the summation of all segment vectors from $u$ to $v$ is fixed, and all grid
	segments have the fixed segment vectors.
Therefore, a path in $\mathcal{C}$ can be characterized by $kd$ 
	integers $x_{11},\ldots,x_{kd}$ satisfying the following for 
	all $1\le i \le d$:
\begin{equation} \label{eqn:req}
\left\{\begin{aligned}
& \begin{aligned}
x_{ji}&\in\{1,\ldots,n'\} && \text{ for }j\in A_i^+, \\
x_{ji}&\in\{-n',\ldots,-1\} && \text{ for }j\in A_i^-, \\
x_{ji}&=0 && \text{ for }j\in A_i^0,
\end{aligned} \\
& \sum_{j\in A_i^+}|x_{ji}|-\sum_{j\in A_i^-}|x_{ji}|=t_i .
\end{aligned}\right.
\end{equation}
The probability that some path exists is the multiplication of the probability of the first edge, the probability of the second edge conditioned on the existence of the first edge, the probability of the third edge conditioned on the existence of the first two edge, etc. In our model, the probability of an (undirected) edge conditioned on the existence of other (undirected) edges is less than or equal to the probability without condition, because each vertex can only 
	connect to exact one other vertex (when considering the edge direction).
Hence we can use the multiplication of the probabilities of all edges as an upper bound of the probability of a path. By union bound and Lemma~\ref{lemma:probedge}, the probability that a path exists in $\mathcal{C}$ is at most
\begin{eqnarray}
& & \sum_{\textrm{all paths in }\mathcal{C}}\quad\prod_{j=1}^k\frac{c_5(\overline{x_{j1}}\cdot\overline{x_{j2}}\cdots\overline{x_{jd}})^{-\frac{\gamma}{d}}}{f(n')} \nonumber \\
& \leq & \frac{c_5^k}{f^k(n')}\quad\sum_{x_{11},\ldots,x_{kd}:\atop \text{satisfying (\ref{eqn:req})}}\quad\prod_{j=1}^k(\overline{x_{j1}}\cdot\overline{x_{j2}}\cdots\overline{x_{jd}})^{-\frac{\gamma}{d}} \nonumber \\
& = & \frac{c_5^k}{f^k(n')}\cdot\prod_{i=1}^d\left\{\sum_{\overline{x_{1i}},\ldots,\overline{x_{ki}}: \sum\limits_{j\in A_i^+}\overline{x_{ji}}-\sum\limits_{j\in A_i^-}\overline{x_{ji}}=t_i; \atop \text{for }j\in A_i^0,\ \overline{x_{ji}}=1\text{; for }j\in A_i^+\cup A_i^-,\ \overline{x_{ji}}\in\{1,\ldots,n'\}}\quad\Big(\prod_{j=1}^k\overline{x_{ji}}\Big)^{-\frac{\gamma}{d}}\right\} \nonumber \\
& = & \frac{c_5^k}{f^k(n')}\cdot\prod_{i=1}^d\Bigg\{\sum_{\overline{x_{ji}}\ (j\text{ ranges in }A_i^+\cup A_i^-): \atop\sum\limits_{j\in A_i^+}\overline{x_{ji}}-\sum\limits_{j\in A_i^-}\overline{x_{ji}}=t_i\text{; and }\overline{x_{ji}}\in\{1,\ldots,n'\}}\quad\Big(\prod_{j\in A_i^+\cup A_i^-}\overline{x_{ji}}\Big)^{-\frac{\gamma}{d}}\Bigg\}. \label{eqn:probc}
\end{eqnarray}
The inner brace summation is considered to be 1 for $A_i^+=A_i^-=\emptyset$. Now we consider the following sum for disjoint sets $A^+,A^-\subseteq\{1,2,\ldots,k\}$ (at least one is not empty) and numbers $t\in\mathbb{Z}$, $\gamma\geq0$, $d\geq0$,
\begin{equation} \label{eqn:sum}
\sum_{y_j\ (j\text{ ranges in }A^+\cup A^-): \atop \sum\limits_{j\in A^+}y_j-\sum\limits_{j\in A^-}y_j=t\text{; and }y_j\in\{1,2,\ldots,n'\}}\quad\left(\prod\limits_{j\in A^+\cup A^-}y_j\right)^{-\frac{\gamma}{d}}.
\end{equation}

\begin{itemize}
\item Case $\gamma=d$.

If one of $A^+$ and $A^-$ is $\emptyset$, we only need to consider the case that $t\neq0$, because the sum~(\ref{eqn:sum}) is 0 for $t=0$. By Lemma~\ref{lemma:calc1}, the sum~(\ref{eqn:sum}) is bounded by \[c_1^k\frac{(\ln|t|)^k}{|t|}\leq c_1^k\max_{x\geq1}\{\frac{(\ln x)^k}{x}\}=(c_1k/e)^k.\]

If neither $A^+$ nor $A^-$ is $\emptyset$, we assume that $t\geq 0$. (The calculation for $t<0$ is similar if we exchange $A^+$ and $A^-$.) By Lemma~\ref{lemma:calc1} the sum~(\ref{eqn:sum}) is bounded by
\begin{eqnarray*}
&& \sum_{s=1}^{kn'}\left(\sum_{\sum\limits_{j\in A^-}y_j=s}\quad\frac{1}{\prod\limits_{j\in A^-}y_j}\right)\left(\sum_{\sum\limits_{j\in A^+}y_j=s+t}\quad\frac{1}{\prod\limits_{j\in A^+}y_j}\right) \\
& \leq & \sum_{s=1}^{kn'}\frac{(c_1\ln s)^{|A^-|-1}}{s}\cdot\frac{(c_1\ln(s+t))^{|A^+|-1}}{s+t} \\
& \leq & c_1^k\sum_{s=1}^{kn'}\frac{(\ln(s+t))^k}{s(s+t)} \\
& = & c_1^k\sum_{s=1}^{kn'}\left(\frac{(\ln(s+t))^k}{\sqrt{s+t}}\cdot\frac{1}{s\sqrt{s+t}}\right) \\
& \leq & c_1^k\cdot\max_{x\geq1}\{\frac{(\ln x)^k}{x^{0.5}}\}\cdot\sum_{s=1}^\infty\frac{1}{s^{1.5}} \\
& = & c_1^k\cdot\frac{(2k)^k}{e^k}\cdot O(1).
\end{eqnarray*}

In either case, the sum~(\ref{eqn:sum}) is bounded by $c_7^kk^k$ for some constant $c_7$. By~(\ref{eqn:probc}) the probability of a path in $\mathcal{C}$ is at most $c_5^k(c_7^kk^k)^d/(\ln n')^k$.

\item Case $0\leq\gamma<d$. By Lemma~\ref{lemma:calc2}, the sum~(\ref{eqn:sum}) is bounded by
\[(c_2n')^{(|A^+|+|A^-|-1)(1-\frac{\gamma}{d})}\leq(c_2n')^{(k-1)(1-\frac{\gamma}{d})}.\]
Therefore, by (\ref{eqn:probc}) the probability of a path in $\mathcal{C}$ is 
	at most
\[
\frac{c_5^k}{(n')^{k(d-\gamma)}}\cdot\left((c_2n')^{(k-1)(1-\frac{\gamma}{d})}\right)^d
	= \frac{c_5^kc_2^{(k-1)(d-\gamma)}}{(n')^{d-\gamma}}.
\]
\end{itemize}
\end{proof}

Finally, we apply Lemmas~\ref{lemma:numcate2},~\ref{lemma:numcate1},
	\ref{lemma:localjump}, and~\ref{lemma:cateogoryexists} together to show that
	the probability of any two vertices connected by a short path with
	at least one long-range edge is diminishingly small.
\begin{lem} \label{lemma:smallprob}
For any two vertices $u$ and $v$, the probability that a path exists connecting
	$u$ and $v$ with at least one long-range edge and
	total length at most $\ell$ is $o(1)$, 
	when $\ell\leq(\log n')^{\frac{1}{1.5(d+1)+\varepsilon}}$ ($\varepsilon>0$) 
	for $\gamma=d$; and $\ell<c \log n'$ 
	($c$ is some constant depending only on $d$ and $\gamma$) 
	for $0\leq\gamma<d$.
\end{lem}
\begin{proof}
We first study the probability of a path with exact length $\ell$,
	and we divide it into the following two cases.
\begin{itemize}
\item
Case $\gamma =d$.
If $k=1$, the probability that a path from $u$ to $v$
	with length $\ell$ exists is at 
	most $O(\ell^dc_6\ln(3\ell)/\ln n')=O(\ell^d\ln \ell/\ln n')$. 
To see this, we divide the path into three segments: 
	a grid segment followed by one long-range edge, then followed
	by another grid segments.
The first grid segment can reach at most $O(\ell^d)$ destinations.
For each of this destination $w$, the long-range edge has to reach
	some vertex within grid distance $\ell$ of vertex $v$.
And by triangle inequality, it must be within grid distance $3\ell$ of vertex $w$ since the distances between $v,u$ and $u,w$ are both at most $\ell$.  
By Lemma~\ref{lemma:localjump}, we know this probability is
	$c_6 \ln (3\ell) / \ln n'$.
Therefore, the above statement holds.

For $k\ge 2$, we simply combine 
	Lemmas~\ref{lemma:numcate1} and~\ref{lemma:cateogoryexists}.
Thus, the probability that a path with length $\ell$ exists is at most
\begin{eqnarray*}
&& O\left(\frac{\ell^d\ln \ell}{\ln n'}\right)+
	\sum_{k=2}^\ell\frac{c_4^k \ell^{(k+1)(d+1)}}{k^{kd}}\cdot
	\frac{c_5^k\left(c_7^kk^k\right)^d}{(\ln n')^k} \\
	&= &
	O\left(\frac{\ell^d\ln \ell}{\ln n'}\right)+
	\sum_{k=2}^\ell\left(\frac{c_4c_5c_7^d
		\cdot \ell^{(d+1)(1+\frac{1}{k})}}{\ln n'}\right)^k \\
 & \le &O\left(\frac{\ell^d\ln \ell}{\ln n'}\right)+
	\sum_{k=2}^\ell\left(\frac{c_4c_5c_7^d
		\cdot \ell^{1.5(d+1)}}{\ln n'}\right)^k .
\end{eqnarray*}

This probability is $o(1)$ when 
	$\ell\leq(\log n')^{\frac{1}{1.5(d+1)+\varepsilon}}$ for any $\varepsilon>0$.
\item
Case $0\le \gamma < d$.
In this case, we combine Lemmas~\ref{lemma:numcate2} 
	and~\ref{lemma:cateogoryexists}.
The probability that a path with length $\ell$ exists is at most
\[{c_3}^\ell\cdot\max_{1\leq k\leq \ell}
	\left\{\frac{c_5^kc_2^{(k-1)(d-\gamma)}}{(n')^{d-\gamma}}\right\}
	\leq \frac{(c_3c_5c_2^{d-\gamma})^\ell}{(n')^{d-\gamma}},\]
where the inequality is based on $c_2,c_5 \ge 1$, which is obviously
	the case.
This probability is $o(1)$ when $\ell<c\log n'$ for 
	some properly chosen constant $c$, which depends only on $d$ and $\gamma$.
\end{itemize}

We now consider the case of path length less than $\ell$.
Let $w$ be a grid neighbor of $v$.
For any path connecting $u$ and $v$ with length less than $\ell$, 
	we can add grid edges $(v,w)$ followed by $(w,v)$, to increase the
	path length to either $\ell-1$ or $\ell$.
Thus, the probability that a path of length at most $\ell$ exists is
	the same as the probability that a path of length $\ell-1$ or
	$\ell$ exists.
The case of $\ell-1$ follows exactly the same as the case of $\ell$
	argued above.
Therefore, the lemma holds.
\end{proof}

With Lemma~\ref{lemma:smallprob}, we are ready to prove 
	Theorem~\ref{thm:main} 
	for the case of $d\ge 2$ and $0\le\gamma \le d$.

\begin{proof}
[Proof of Theorem~\ref{thm:main} for the case of $d\ge 2$ and $0\le \gamma \le d$.]
We define $\ell = \lfloor (\log n')^{\frac{1}{1.5(d+1)+\varepsilon}}/2 \rfloor$ when
	$\gamma=d$, and
	$\ell = \lfloor c \log n'/2 \rfloor$ when $0\le \gamma <d$, where
	$c$ is the constant determined by Lemma~\ref{lemma:smallprob}.
In the base grid, find any square in any dimension with one side length
	equal to $\ell$.
Consider the four corner vertices of the square.
By Lemma~\ref{lemma:smallprob}, the probability that
	a pair of corners of this square are connected by a 
	path with at least one long-range edge and total
	length at most $2\ell$ is $o(1)$.
Thus by union bound, the distances between any pair of these four corner
	vertices are exactly their grid distance, with probability $1-o(1)$.
This means that the $\delta$ value of these four vertices is $\ell$, 
	with probability $1-o(1)$.
Therefore, we know that with probability $1-o(1)$, 
	$\delta(\KSW(n, d, \gamma)) = \Omega((\log n')^{\frac{1}{1.5(d+1)+\varepsilon}}) = \Omega((\log n)^{\frac{1}{1.5(d+1)+\varepsilon}})$
	when $d \ge 2$ and $\gamma =d$, and
	$\delta(\KSW(n, d, \gamma)) = \Omega(\log n') = \Omega(\log n)$ when
	$d \ge 2$ and $0\le \gamma < d$.
\end{proof}

\noindent
{\bf Remark (The limitation of this approach).}
\ \ 
We already have tight lower bound for $0\leq\gamma<d$. 
However, the lower bound $\Omega((\log n)^{\frac{1}{1.5(d+1)+\varepsilon}})$ for $\gamma=d$ 
	is still not matching the upper bound $O(\log n)$. 
We show below that the lower bound can never be improved to $\Omega(\log n)$ 
	by the above technique of 
	proving that the grid distance is the shortest on every small square 
	(with high probability).

Consider a \KSW graph with $\gamma=d$. 
For an $\ell\times \ell$ square $S$, let $u$ be its the upper left vertex 
	and $v$ be its lower right vertex. 
Similar to the proof in Lemma~\ref{lemma:localjump}, the probability of 
	existence of an edge linking any $w$ into the $\ell/2\times \ell/2$ square 
	on the lower right side of $w$ within the square $S$ (but
	excluding $w$'s two grid neighbors in the square) is 
\[q=\Theta\left(\sum_{i=2}^{\ell/2}i\cdot\frac{i^{-d}}{\ln n'}\right)=
	\begin{cases}\Theta(\frac{\log \ell}{\log n'}) & d=2,\\ \Theta(\frac{1}{\log n'}) & d\geq3.\end{cases}\]
Consider the $\ell/2\times \ell/2$ square on the lower right side of $u$, which
	is the upper left quadrant of the original square. 
The probability that at least one vertex $w$ in this quadrant links 
	to its lower right $\ell/2\times \ell/2$ square is 
	$1-(1-q)^{\ell/2\times \ell/2}$. 
This probability is almost $1$ when 
	$\ell = \omega\left(\sqrt{\frac{\log n'}{\log\log n'}}\right)$ for $d=2$, 
	or $\ell = \omega(\sqrt{\log n'})$ for $d\geq3$. 
If there exists such a vertex $w$ in the upper right quadrant, 
	and suppose it links to a vertex $x$ in its lower right 
	$\ell/2\times \ell/2$ square.
Then $x$ must also be in the original square $S$, and we can have a path from 
	$u$ to $w$ following the grid path, then the long-range edge from 
	$w$ to $x$, and then the grid path from $x$ to $v$.
This path connects $u$ and $v$ and must be shorter than the grid paths
	from $u$ to $v$.
Therefore our technique cannot improve the lower bound to
	$\omega\left(\sqrt{\frac{\log n'}{\log\log n'}}\right)=\omega\left(\sqrt{\frac{\log n}{\log\log n}}\right)$ for $d=2$ 
	or $\omega(\sqrt{\log n'})=\omega(\sqrt{\log n})$ for $d\ge 3$.

\subsubsection{The case of $d=1$ and $0 \le \gamma \le 1$} \label{sec:1dsmallr}

In this section we give lower bounds of $\delta$ for the one dimensional
	\KSW model (based on an $n$-vertex ring). 
Let the $n$ vertices be $v_0,\ldots,v_{n-1}$. Let $\ell_0=\lfloor(\log n)^{\frac{1}{1.5(d+1)+\varepsilon}}\rfloor$ when $\gamma=1$, or $\lfloor c\log n\rfloor$ when $0\leq\gamma<1$, where $c$ is the number determined in Lemma~\ref{lemma:smallprob}.

The idea is to find a long-range edge $e_0$ between two vertices with grid distance $\ell_0$. 
As $e_0$ forms a local ring with the original grid, we can give a lower bound 
	of $\delta$ such that the ring distances (with respect to the 
	grid edges and $e_0$) are the shortest even after adding the long-range
	edges. 
We first calculate the probability of ring distances being the shortest under the condition of existing $e_0$.

We divide the construction of a \KSW graph into two stages: 1) Every vertex 
	links to exactly one other vertex according to some distribution; and
	2) we ignore the edge direction and consider the graph as undirected. 
Let ${\cal E}_i$ ($0\leq i\leq n-1$) be the event that $v_i$ links to 
	$v_{(i+\ell_0\bmod n)}$ in the first stage. 
Under the condition that ${\cal E}_i$ happens, $v_{(i+\ell_0\bmod n)}$ is still free to link to any vertex. 
The events ${\cal E}_0,{\cal E}_1,\ldots,{\cal E}_{n-1}$ are independent.

\begin{lem} \label{lemma:1dsmallprob}
Under the condition that ${\cal E}_i$ happens ($v_i$ links to $v_{(i+\ell_0\bmod n)}$), for any two vertices $u,w$ on the curve from $v_i$ to $v_{(i+\ell_0\bmod n)}$ (exclusive), the conditional probability of existing a path connecting $u$ and $w$ with at least one long-range edge other than $(v_i,v_{(i+\ell_0\bmod n)})$ and path length at most $\ell_0$ is $o(1)$.
\end{lem}

\begin{proof}
Let $e_0$ be the edge $(v_i,v_{(i+\ell_0\bmod n)})$.
We still follow the arguments for $d\geq2$ and $0\leq\gamma\leq d$, but change the classification of paths slightly: now $e_0$ is a type of edges by itself,
	and thus together with grid edges and long-range edges, we have
	three types of edges and three types of corresponding segments.
For each original category of paths with length $\ell$ and $k$ long-range
	segments defined in the previous section, we further
	divide it into at most $k+1$ categories, based on whether 
	$e_0$ was the first, second, ..., or the $k$-th long-range segment, or
	$e_0$ does not appear in the path.
The number of long-range edges in a category with $e_0$ is decreased by 1. 
For the categories with only grid edges and $e_0$, those paths now have no 
	long-range edge and they exist with conditional probability 1, 
	hence they will not be calculated. 

For Lemma~\ref{lemma:numcate2}, the number of categories is increased by a 
	factor of at most $\ell+1$.
For Lemma~\ref{lemma:numcate1}, the number of categories is increased by a 
	factor of at most $k+1$.
Thus, we only need to properly adjust the constants $c_3$ and $c_4$ in these
	two lemmas to make them still hold.
Lemma~\ref{lemma:probedge} is not changed except that $e_0$ is no longer a long-range edge and the lemma is not applicable on $e_0$. Lemma~\ref{lemma:localjump} is not affected. In the proof of Lemma~\ref{lemma:cateogoryexists}, a category without $e_0$ can be calculated normally. For a category with $e_0$, the number of long-range edges is decreased by 1 as stated above. One can see that 
all the arguments still hold by changing the summation of long-range edges ($t_i$ in~(\ref{eqn:req})) to contain $e_0$.

For Lemma~\ref{lemma:smallprob}, only the case $\gamma=d$ and $k=1$ (one long-range edge) needs some modification. If the path does not contain $e_0$, the calculation still holds. Otherwise, we can divide the path into five segments: grid, long-range, grid, $e_0$, grid (or grid, $e_0$, grid, long-range, grid). We just consider the 3 consecutive segments grid, $e_0$, grid as a whole, which contains $O(\ell)$ edges and can reach $O(\ell^d)=O(\ell)$ destinations. One can see the previous argument still holds. Therefore, the consequence of Lemma~\ref{lemma:smallprob} still holds.
\end{proof}

With Lemma~\ref{lemma:1dsmallprob}, we give the proof of Theorem~\ref{thm:main} for the case of $d=1$ and $0\leq\gamma \le 1$.

\begin{proof}
[Proof of Theorem~\ref{thm:main} for the case of $d=1$ and $0\le \gamma \le 1$.]
Let ${\cal F}$ be the event $\delta\geq\frac{1}{4}\ell_0-3$. We first show that $\Pr\{{\cal F}\mid {\cal E}_i\}=1-o(1)$.
Pick four vertices $A_i=v_{(i+\lfloor\frac{1}{8}\ell_0\rfloor\bmod n)}$, $B_i=v_{(i+\lfloor\frac{3}{8}\ell_0\rfloor\bmod n)}$, $C_i=v_{(i+\lfloor\frac{5}{8}\ell_0\rfloor\bmod n)}$ and $D_i=v_{(i+\lfloor\frac{7}{8}\ell_0\rfloor\bmod n)}$. 
By Lemma~\ref{lemma:1dsmallprob} and union bound, 
	we know that with probability $1-o(1)$, the distances between every pair 
	of vertices are the ring (grid edges plus $(v_i,v_{(i+\ell_0\bmod n)})$) distances. That is,
	the distances between $A_i B_i$, $B_i C_i$, $C_i D_i$, $D_i A_i$ are 
	roughly $\frac{1}{4}\ell_0$ (off by at most $2$), and 
	the distances between $A_i C_i$ and $B_i D_i$ are roughly 
	$\frac{1}{2}\ell_0$ (off by at most $1$). 
Therefore, by considering the four vertices $A_i B_i C_i D_i$, 
	we have $\delta$ at least $\frac{1}{2}(\frac{1}{2}\ell_0-1+\frac{1}{2}\ell_0-1-\frac{1}{4}\ell_0-2-\frac{1}{4}\ell_0-2)=\frac{1}{4}\ell_0-3$ with 
	conditional probability $1-o(1)$.

For every ${\cal E}_i$, the probability that $v_i$ links to $v_{(i+\ell_0\bmod n)}$ is \[\Pr\{{\cal E}_i\}=\Theta(\frac{\ell_0^{-\gamma}}{f(n)})=\begin{cases}\Theta((\log n)^{-\frac{2+\varepsilon}{3+\varepsilon}}) & \gamma=1, \\ \Theta(\log n/n^{1-\gamma}) & 0\leq\gamma<1.\end{cases}\] Define this probability as $q$. We have $\Pr\{{\cal F}\text{ and }{\cal E}_i\}=\Pr\{{\cal F}\mid{\cal E}_i\}\Pr\{{\cal E}_i\}=(1-o(1))q$.

Let $K$ be the random variable denoting the number of ${\cal E}_i$'s that 
	occur.
We define $m=E[K]$, and we have $m=nq$.
One can check that $m=poly(n)$ for both cases $\gamma=1$ and $0\leq\gamma<1$. By Chernoff Bound, $\Pr\{|K-m|\leq m^{0.6}\}>1-2e^{-(m^{-0.4})^2/4\cdot m}=1-2e^{-m^{0.2}/4}$. 
Hence with very high probability, $K$ is close to $m$. Let ${\cal G}$ denote the event that $m-m^{0.6}\leq K \leq m+m^{0.6}$. We show that $\Pr\{{\cal F}\text{ and }{\cal E}_i\text{ and }{\cal G}\}$ is very close to $\Pr\{{\cal F}\text{ and }{\cal E}_i\}$. In fact, $\Pr\{{\cal F}\text{ and }{\cal E}_i\text{ and }{\cal G}\}\geq\Pr\{{\cal F}\text{ and }{\cal E}_i\}-\Pr\{\text{not }G\}=(1-o(1))\Pr\{{\cal F}\text{ and }{\cal E}_i\}$, because $\Pr\{\text{not }{\cal G}\}<2e^{-m^{0.2}/4}=2e^{-poly(n)}$ is much smaller than $\Pr\{{\cal F}\text{ and }{\cal E}_i\}=(1-o(1))q=\omega(1/n)$. On the other hand, it is straightforward that $\Pr\{{\cal F}\text{ and }{\cal E}_i\text{ and }{\cal G}\}\leq\Pr\{{\cal F}\text{ and }{\cal E}_i\}$. Hence $\Pr\{{\cal F}\text{ and }{\cal E}_i\text{ and }{\cal G}\}$ differs from $\Pr\{{\cal F}\text{ and }{\cal E}_i\}$ by a factor at most $(1-o(1))$, and $\Pr\{{\cal F}\text{ and }{\cal E}_i\text{ and }{\cal G}\}=(1-o(1))q$.

For every $r=r_0r_1\cdots r_{n-1}\in\{0,1\}^n$, define ${\cal H}_r$ to be the event that ${\cal E}_j$ happens iff $r_j=1$ for all $j=0,1,\ldots,n-1$. We can see ${\cal H}_r$'s are mutually exclusive for $r\in\{0,1\}^n$. Hence \[\Pr\{{\cal F}\text{ and }{\cal E}_i\text{ and }{\cal G}\}=\sum_{k=m-m^{0.6}}^{m+m^{0.6}}\sum_{r\in\{0,1\}^n\atop r_i=1\text{ and }r\text{ has $k$ $1$'s}}\Pr\{{\cal F}\text{ and }{\cal H}_r\}.\] Therefore
\begin{align*}
\Pr\{{\cal F}\text{ and }{\cal G}\}&=\sum_{k=m-m^{0.6}}^{m+m^{0.6}}\sum_{r\in\{0,1\}^n\atop r\text{ has $k$ $1$'s}}\Pr\{{\cal F}\text{ and }{\cal H}_r\}\\
&=\sum_{k=m-m^{0.6}}^{m+m^{0.6}}\frac{1}{k}\cdot\sum_{i=0}^{n-1}\sum_{r\in\{0,1\}^n\atop r_i=1\text{ and }r\text{ has $k$ $1$'s}}\Pr\{{\cal F}\text{ and }{\cal H}_r\} \\
&>\frac{1}{m+m^{0.6}}\cdot\sum_{i=0}^{n-1}\sum_{k=m-m^{0.6}}^{m+m^{0.6}}\sum_{r\in\{0,1\}^n\atop r_i=1\text{ and }r\text{ has $k$ $1$'s}}\Pr\{{\cal F}\text{ and }{\cal H}_r\} \\
&=\frac{1}{m+m^{0.6}}\sum_{i=0}^{n-1}\Pr\{{\cal F}\text{ and }{\cal E}_i\text{ and }{\cal G}\} \\
&=\frac{n(1-o(1))q}{m+m^{0.6}}=\frac{m(1-o(1))}{m+m^{0.6}}=1-o(1).
\end{align*}
The probability that $\delta\geq\frac{1}{4}\ell_0-3=\Omega(\ell_0)$ is $\Pr\{{\cal F}\}\geq\Pr\{{\cal F}\text{ and }{\cal G}\}=1-o(1)$.
\end{proof}

\subsubsection{The case of $d=1$ and $\gamma > 3$} \label{sec:1dlarger}

We first show that with high probability all long-range edges connect
	two vertices with grid distance $o(n)$, for general $d$ and $\gamma>2d$.

\begin{lem} \label{lemma:shortedge}
In a random graph from $\KSW(n,d,\gamma)$ with $\gamma>2d$,	
	with probability $1-o(1)$
	there is no long-range edge that
	connects two vertices with grid distance
	larger than $n^{\frac{1}{\gamma-d}+\varepsilon}$, where $\varepsilon$ is any positive number.
\end{lem}

\begin{proof}
For a vertex $u$, the probability that the long-range edge from $u$ links to somewhere with distance longer than 
	$\ell_0=n^{\frac{1}{\gamma-d}+\varepsilon}$ from $u$ is 
\[O\left(\sum_{i=\ell_0}^{n^{1/d}}i^{d-1}\frac{i^{-\gamma}}{\sum_{j=1}^{n^{1/d}}j^{d-1}j^{-\gamma}}\right)=O\left(\sum_{i=\ell_0}^\infty i^{d-1-\gamma}\right)=O\left(\ell_0^{d-\gamma}\right).\] By union bound, the probability that such $u$ exists is 
\[O\left(n\ell_0^{d-\gamma}\right)=O\left(n\cdot n^{(\frac{1}{\gamma-d}+\varepsilon)(d-\gamma)}\right)=O\left(n^{\varepsilon(d-\gamma)}\right)=o(1).\] 
Therefore with probability $1-o(1)$ such $u$ does not exist.
\end{proof}

Given a graph $G$ in $\KSW(n,d,\gamma)$, 
	let $\ell_0(G)$ be the largest grid distance of two vertices connected by 
	a long-range edge in $G$.
From the above result, we know that when $\gamma>3d$, 
	$\ell_0(G)<n^{\frac{1}{\gamma-d}+\varepsilon}=o(\sqrt{n^{1/d}})$ with high probability.
For the rest of this section, with $d=1$, 
	we fix $G$ to be any graph in $\KSW(n,1,\gamma)$
	with $\ell_0(G)<n^{\frac{1}{\gamma-1}+\varepsilon}$, and 
	show that $\delta(G) = \Omega(n^c)$
	for some constant $c$.
Since $G$ is fixed, we will use $\ell_0$ to be the short hand of $\ell_0(G)$.


Now go back to the one dimensional grid with wrap-around, which is
	a ring with vertices $v_0,v_1,\ldots, v_{n-1}$.
Notice that in one-dimensional case, the edge vector defined in 
	Section~\ref{sec:d2} degenerates to a scalar value from
	$\{-\lfloor\frac{n}{2}\rfloor,-\lfloor\frac{n}{2}\rfloor+1,
	\ldots,\lfloor\frac{n-1}{2}\rfloor\}$.
We arrange $v_0,v_1,\ldots, v_{n-1}$ clockwise on the ring.
Then a positive edge scalar corresponds to a clockwise hop while a negative
	edge scalar corresponds to a counter-clockwise~hop.

Let $A=v_0$ and $B=v_{\lfloor n/2\rfloor}$ be two specific vertices. 
We define two kinds of paths between $A$ and $B$: a {\em positive path}
	is one in which the summation of edge scalars (not taking module $n$)
	is positive, while a {\em negative path} is one in which the summation
	of edge scalars (not taking module $n$) is negative.





\begin{lem} \label{lemma:ABshort}
There exists a positive path from $A$ to $B$ that does not go through $v_{\lfloor n/2\rfloor+1}$, $v_{\lfloor n/2\rfloor+2}$, \ldots,$v_{n-1}$, and the length is 
	at most $2\ell_0$ longer than the shortest positive path from $A$ to $B$.
Similarly, there exists a negative path from $A$ to $B$ that does not go through $v_1,v_2,\ldots,v_{\lfloor n/2\rfloor-1}$, and the length is at most $2\ell_0$ longer than the shortest negative path from $A$ to $B$.
\end{lem}
\begin{proof}
We only give the proof for the positive path case. 
Consider any shortest positive path $\cal P$ from $A$ to $B$. 
We first show the following claim.
Let $S_A = \{v_1,\ldots,v_{\ell_0}\}$ be the set of 
	$\ell_0$ consecutive vertices clockwise to $A$, and
	$S_B  = \{v_{\lfloor n/2\rfloor-\ell_0},v_{\lfloor n/2\rfloor-\ell_0+1},
	\ldots,v_{\lfloor n/2\rfloor-1}\}$ be the set of $\ell_0$ consecutive
	vertices counter-clockwise to $B$.

{\em Claim.} There must exist a subpath in $\cal P$ from 
	a vertex $u \in S_A$ to a vertex $w \in S_B$ that does not go 
	through $v_{\lfloor n/2\rfloor+1}$, $v_{\lfloor n/2\rfloor+2}$, \ldots,$v_{n-1}$. 


Say there are $m$ edges in $\cal P$. Let $s_i$ ($0\le i\le m$) denote the summation of the first $i$ edge scalars in $\cal P$ (not taking module $n$). Initially we have $s_0=0$, and the final value $s_m$ is $kn+\lfloor n/2\rfloor$ for some integer $k\ge 0$. One can also see that the position after going through the first $i$ edges in $\cal P$ is at $v_{(s_i\bmod n)}$. Let $i_1$ be the smallest integer such that $s_{i_1}\geq\lfloor n/2\rfloor$ (exist because $s_m\geq\lfloor n/2\rfloor$), and $i_2$ be the largest integer such that $i_2<i_1$ and $s_{i_2}\leq0$ (exist because $s_0=0$).

We consider the $(i_2+1)$-th edge, which begins at $v_{(s_{i_2}\bmod n)}$ for some $s_{i_2}\leq 0$, and ends at $v_{(s_{i_2+1}\bmod n)}$ for some $s_{i_2+1}>0$. The number $s_{i_2+1}$ is at most $s_{i_2}+\ell_0$ since no edge is longer than $\ell_0$. Therefore, $s_{i_2+1}$ must be a number in $(0,\ell_0]$ and $v_{(s_{i_2+1}\bmod n)}$ must be in $S_A$. We choose $u=v_{(s_{i_2+1}\bmod n)}\in S_A$. Similarly, pick $w=s_{i_1-1}$, which is the beginning point of the $i_1$-th edge, we have $w\in S_B$. The intermediate values $s_{i_2+1},s_{i_2+2},\ldots,s_{i_1-1}$ are all in the interval $(0,\lfloor n/2\rfloor)$ by the definitions of $i_1$ and $i_2$. 
That is, for all $j$ such that $i_2 + 1 \le j \le i_1-1$, 
	$s_j \bmod n = s_j$, and the corresponding vertex $v_{(s_j\bmod n)} = v_{s_j}
	\in \{v_1, v_2, \ldots, v_{\lfloor n/2\rfloor-1}\}$.
Therefore the subpath from $u$ to $w$ does not go through $v_{\lfloor n/2\rfloor+1}$, $v_{\lfloor n/2\rfloor+2}$, \ldots,$v_{n-1}$, and the claim holds.

With the claim, we can construct a positive path $\cal P'$, which
	use ring edges from $A$ to $u$, then use the subpath in the claim
	from $u$ to $w$, and then from $w$ to $B$ using ring edges. 
The length of $\cal P'$ is at most $2\ell_0$ longer than $\cal P$, 
	the shortest positive path from $A$ to $B$.
\end{proof}

We use ${\cal P}_{AB}^+$ and ${\cal P}_{AB}^-$ to denote the two paths 
	stated in the above lemma. According to this lemma, one of ${\cal P}_{AB}^+$ and ${\cal P}_{AB}^-$ is at most $2\ell_0$ longer than the shortest path between $A$ and $B$.

Let $C$ be the middle point of ${\cal P}_{AB}^+$ (take a vertex nearest middle if the path has odd number of edges), and ${\cal P}_{AC}^+$, ${\cal P}_{CB}^+$ be the two subpaths from $A$ to $C$ and $C$ to $B$. Similarly, let $D$ be the middle point of ${\cal P}_{AB}^-$ and ${\cal P}_{AD}^-$, ${\cal P}_{DB}^-$ be the two subpaths.

\begin{lem} \label{lemma:ACshort}
The paths ${\cal P}_{AC}^+$, ${\cal P}_{CB}^+$, ${\cal P}_{AD}^-$ and ${\cal P}_{DB}^-$ are at most $3\ell_0+1$ longer than the shortest paths between corresponding pairs of vertices.
\end{lem}

\begin{proof}
We only give the proof for ${\cal P}_{AC}^+$. Suppose that it is not true, and the shortest path ${\cal P}_{AC}^*$ from $A$ to $C$ is at least $3l_0+2$ shorter than ${\cal P}_{AC}^+$. The path ${\cal P}_{AC}^*$ must be at least $3\ell_0+1$ shorter than ${\cal P}_{CB}^+$ because $C$ is the point nearest middle of ${\cal P}_{AB}^+$.

Consider the last time that the path ${\cal P}_{AC}^*$ gets into the range $\{v_0, v_1, \ldots, v_{\lfloor n/2\rfloor}\}$, the subpath of ${\cal P}_{AC}^*$ 
	from that point to $C$ must be one of the following cases.

\begin{itemize}
\item It is a subpath from a vertex $A'\in\{v_0,v_1,\ldots,v_{\ell_0}\}$ to $C$ not going through $v_{\lfloor n/2\rfloor+1}$, $v_{\lfloor n/2\rfloor+2}$, \ldots,$v_{n-1}$. Replace ${\cal P}_{AC}^+$ by the path from $A$ to $A'$ through ring edges concatenated with the subpath of ${\cal P}_{AC}^*$ 
	from $A'$ to $C$. This will cause the length of ${\cal P}_{AB}^+$ to decrease by at least $3\ell_0+2-\ell_0>2\ell_0$, which is impossible by Lemma~\ref{lemma:ABshort}.

\item It is a subpath from a vertex $B'\in\{v_{\lfloor n/2\rfloor-\ell_0},v_{\lfloor n/2\rfloor-\ell_0+1},\ldots,v_{\lfloor n/2\rfloor}\}$ to $C$ that does not go through $v_{\lfloor n/2\rfloor+1}$, $v_{\lfloor n/2\rfloor+2}$, \ldots,$v_{n-1}$. Replace ${\cal P}_{CB}^+$ by the reverse of this subpath of ${\cal P}_{AC}^*$ from $C$ to $B'$ 
	concatenated with ring edges from $B'$ to $B$. 
This will cause the length of ${\cal P}_{AB}^+$ to decrease by at least $3\ell_0+1-\ell_0>2\ell_0$, which is impossible by Lemma~\ref{lemma:ABshort}.
\end{itemize}
Therefore the lemma holds.
\end{proof}

Then we consider the shortest path between $C$ and $D$.

\begin{lem}
Either the concatenation of ${\cal P}_{CB}^+$ and reversed ${\cal P}_{DB}^-$, or the concatenation of reversed ${\cal P}_{AC}^+$ and ${\cal P}_{DA}^-$ is at most $8\ell_0+2$ longer than the shortest path between $C$ and $D$.
\end{lem}

\begin{proof}
The shortest path from $C$ to $D$ (say ${\cal P}_{CD}^*$) must go through either $B$'s neighborhood $v_{\lfloor n/2\rfloor}$, $v_{\lfloor n/2\rfloor+1}$, \ldots, $v_{\lfloor n/2\rfloor+\ell_0}$ or $A$'s neighborhood $v_0,v_1,\ldots,v_{\ell_0}$. Without loss of generality, we assume that it goes through the point $B'$ in $B$'s neighborhood. Use ${\cal P}_{CB'}^*$ and ${\cal P}_{B'D}^*$ to denote the two subpaths from $C$ to $B'$ and $B'$ to $D$ respectively. They must also be shortest paths of $CB'$ and $B'D$.

The path ${\cal P}_{CB'}^*$ is at most $\ell_0$ shorter than the shortest path between $C$ and $B$, otherwise the path ${\cal P}_{CB'}^*$ concatenated with ring edges from $B'$ to $B$ would be shorter than the shortest path. Similarly, ${\cal P}_{B'D}^*$ is at most $\ell_0$ shorter than the shortest path between $B$ and $D$. Therefore the shortest path between $C$ and $D$ is at most $2\ell_0$ shorter than the concatenation of shortest paths of $CB$ and $BD$. Then by Lemma~\ref{lemma:ACshort}, the summation of ${\cal P}_{CB}^+$ and ${\cal P}_{DB}^-$ is at most $2(3\ell_0+1)+2\ell_0=8\ell_0+2$ longer than the shortest path between $C$ and $D$.
\end{proof}

We have the following corollary since the two paths in this lemma differ by at most 2 considering the length.

\begin{coro} \label{cor:CDshort}
The concatenation of ${\cal P}_{CB}^+$ and reversed ${\cal P}_{DB}^-$, and the concatenation of reversed ${\cal P}_{AC}^+$ and ${\cal P}_{DA}^-$ are both at most $8\ell_0+4$ longer than the shortest path between $C$ and $D$.
\end{coro}

Now we can prove the lower bound of $\delta$ for this case.

\begin{proof}
[Proof of Theorem~\ref{thm:main} for the case of $d=1$ and $\gamma>3$.]
Consider the four points $A$, $B$, $C$ and $D$ defined above. 
Let $d(x,y)$ denote the distance between vertices $x$ and $y$.
We can see the following consequences about pairwise distances: 
	(a) $d(A,B) \ge \lfloor n/2\rfloor/\ell_0$; 
	(b) $d(C,D) \ge |{\cal P}_{CB}^+| +  |{\cal P}_{DB}^-| - (8\ell_0+4)
	\ge |{\cal P}_{AC}^+|-1 +  |{\cal P}_{DB}^-| - (8\ell_0+4)
	\ge d(A,C) + d(D,B) - (8\ell_0+5)$, where the first inequality is
	due to Corollary~\ref{cor:CDshort}; and
	(c) similarly, $d(C,D) \ge d(A,D) + d(C,B) - (8\ell_0+5)$.
Therefore, we have both 
	$d(A,B) + d(C,D) \ge  d(A,C) + d(D,B) + \lfloor n/2\rfloor/\ell_0 
	- (8\ell_0+5)$, and 
	$d(A,B) + d(C,D) \ge  d(A,D) + d(C,B) + \lfloor n/2\rfloor/\ell_0 
	- (8\ell_0+5)$.

For $\ell_0<n^{\frac{1}{\gamma-1}+\varepsilon}$ with any sufficiently small
	$\varepsilon > 0$ and sufficiently large $n$, we have 
	$\lfloor n/2\rfloor/\ell_0>>8\ell_0+5$,  
	and thus $d(A,B) + d(C,D)$ is the largest
	distance pair. 
In this case, $\delta\geq\lfloor n/2\rfloor/\ell_0-(8\ell_0+5)$. By Lemma~\ref{lemma:shortedge}, with probability $1-o(1)$ there is $\ell_0<n^{\frac{1}{\gamma-1}+\varepsilon}$. Therefore, with probability $1-o(1)$, 
	$\delta(\KSW(n,1,\gamma))=\Omega(n/n^{\frac{1}{\gamma-1}+\varepsilon})=
	\Omega(n^{\frac{\gamma-2}{\gamma-1}-\varepsilon})$ for $d=1$, $\gamma>3$ and 
	any sufficiently small $\varepsilon>0$.
Since for any $\varepsilon' > \varepsilon > 0$, 
	$n^{\frac{\gamma-2}{\gamma-1}-\varepsilon} = 
	\Omega(n^{\frac{\gamma-2}{\gamma-1}-\varepsilon'})$, we have
	$\delta(\KSW(n,1,\gamma))= \Omega(n^{\frac{\gamma-2}{\gamma-1}-\varepsilon})$
	for any $\varepsilon > 0$.
%
%
%
%
\end{proof}

\subsection{Extensions of the \KSW model} \label{sec:sworldextension}

Our analysis also holds for some variants of the \KSW model. In this section, we study one variant of the underlying structure: grid without wrap-around; and two variants of edge linking: multiple edges for each vertex and linking edges independently.

\vspace{\topsep}

\noindent {\bf Grid without wrap-around.} We modify our analysis so that the first two results of Theorem~\ref{thm:main} still hold. For the case of $d\geq2$ and $0\leq\gamma\leq d$ (Section~\ref{sec:d2}), the changes are as follows. For a path from $u$ to $v$, we divide it into segments as before. Elements in an edge vector are in $\{-n',-n'+1,\ldots,n'-1,n'\}$ now. Recall the last condition that we define two paths from $u$ to $v$ belong to the same category: the summations (not module $n'$) of all segment vectors in the two paths are equal. This is always satisfied for grid without wrap-around, because the summation of all segment vectors depends only on the positions of $u$ and $v$. In the proofs of Lemma~\ref{lemma:numcate2} and Lemma~\ref{lemma:numcate1}, the summation of all segment vectors is fixed rather than $(2\ell+1)^d$ or $(2k+1)^d$ choices respectively. Hence the upper bounds given in Lemmas~\ref{lemma:numcate2} and~\ref{lemma:numcate1} still hold. In Lemma~\ref{lemma:probedge}, an edge between $u$ and $v$ can be from $u$ to $v$ or from $v$ to $u$. The probabilities of the two cases may differ by a constant factor on grid without wrap-around. Hence the probability of existing an edge between $u$ and $v$ is changed by at most a constant factor, and Lemma~\ref{lemma:probedge} still holds. One can verify the rest analysis in Section~\ref{sec:d2} still hold. For the case of $d=1$ and $0\leq\gamma\leq1$ (Section~\ref{sec:1dsmallr}), the only change is that event  ${\cal E}_i$, which 
	is the event that $v_i$ links to $v_{i+\ell_0}$, only applies when
	$i=0,1,\ldots,n-\ell_0-1$. 
Since $\ell_0=O(\log n)$ is much smaller than $n$, there are still almost $n$ events ${\cal E}_i$ and the argument has no significant change. 
Hence Theorem~\ref{thm:main} still holds for the cases that $d\geq1$ and $0\leq\gamma\leq d$.

\vspace{\topsep}

\noindent {\bf Multiple edges for each vertex.} In this model, each vertex links a constant, say $d_0$, number of edges according to the same distribution that $u$ links to $v$ with probability $\frac{d_B(u,v)^{-\gamma}}{\sum_{v'}d_B(u,v')^{-\gamma}}$. We show that all our analysis still hold with some slight changes. For the case of $d\geq2$ and $1\leq\gamma\leq d$ (Section~\ref{sec:d2}), Lemma~\ref{lemma:probedge} still holds since the probability of the edge $(u,v)$ is increased by at most $d_0$ times using union bound. And the proof of Lemma~\ref{lemma:localjump} still works, because $O(i^{d-1}\frac{i^{-\gamma}}{f(n')})$ is an upper bound of the probability that $u$ links to some vertex at distance $i$ by union bound. For the case of $d=1$ and $0\leq\gamma\leq1$ (Section~\ref{sec:1dsmallr}), we define ${\cal E}_i$ as the event that at least one of $v_i$'s edges links to $v_{(i+\ell_0\bmod n)}$. One can see $\Pr\{{\cal E}_i\}$ is still $\Theta(\frac{\ell_0^{-\gamma}}{f(n)})$ and the rest argument also holds. For the case of $d=1$ and $\gamma>3$ (Section~\ref{sec:1dlarger}), Lemma~\ref{lemma:shortedge} still holds for the same reason as Lemma~\ref{lemma:localjump}, that $O(i^{d-1}\frac{i^{-\gamma}}{\sum_{j=1}^{n^{1/d}}j^{d-1}j^{-\gamma}})$ is an upper bound of the probability that $u$ links to some vertex at distance $i$. One can verify all results of Theorem~\ref{thm:main} still hold under this change.

\vspace{\topsep}

\noindent {\bf Linking edges independently.} In this model, all edges exist independently. The edge between $(u,v)$ exists with probability $\frac{d_0\cdot d_B(u,v)^{-\gamma}}{\sum_{i=1}^{n^{1/d}}i^{d-1-\gamma}}$, where $d_0$ is some constant. We also give the changes in our analysis. Lemma~\ref{lemma:probedge} is straightforward in this model. Lemmas~\ref{lemma:localjump} and~\ref{lemma:shortedge} still hold for the same reason as above. In the proof of Lemma~\ref{lemma:cateogoryexists}, the first line of Eq~(\ref{eqn:probc}), which uses the multiplication of edges' probabilities for an upper bound of the path's probability, still holds because it is now just the multiplication of independent events. For the case of $d=1$ and $0\leq\gamma\leq1$ (Section~\ref{sec:1dsmallr}), we define ${\cal E}_i$ to be the event that the edge $(v_i,v_{(i+\ell_0\bmod n)})$ exists, one can see the analysis still works. All results of Theorem~\ref{thm:main} still hold under this change.

\vspace{\topsep}

In summary, Theorem~\ref{thm:main} of the case $d\geq1$ and $0\leq\gamma\leq d$ still holds for grid without wrap-around, and Theorem~\ref{thm:main} of all cases still holds for both variants of edge linking. The variants of edge linking can be combined with grid without wrap-around, for which Theorem~\ref{thm:main} of the case $d\geq1$ and $0\leq\gamma\leq d$ still holds.

\section{$\delta$-hyperbolicity of ringed trees} \label{sec:ringedtree}

In this section, we consider the $\delta$-hyperbolicity of graphs constructed 
according to a variant of the small-world graph model, in which long-rang 
edges are added on top of a base graph that is a binary tree or tree-like 
low-$\delta$ graph.
In particular, we will analyze the effect on the $\delta$-hyperbolicity of 
adding long-range links to a ringed tree base graph; and then we will 
consider several related extensions, including an extension to the binary 
tree.

\begin{defi}[Ringed tree]
A {\em ringed tree} of level $k$, denoted $RT(k)$, is a fully binary tree 
with $k$ levels (counting the root as a level), in which all vertices at the 
same level are connected by a ring.
More precisely, we can use a binary string to represent each vertex
	in the tree, such that the root (at level $0$) 
	is represented by an empty string,
	and the left child and the right child of a vertex with string $\sigma$ are
	represented as $\sigma 0$ and $\sigma 1$, respectively.
Then, at each level $i=1,2,\ldots, k-1$, we connect two vertices
	$u$ and $v$ represented by binary strings $\sigma_u$ and $\sigma_v$
	if $(\sigma_u + 1) \mod 2^i = \sigma_v$, where the addition treats
	the binary strings as the integers they represent. 
As a convention, we say that a level is higher if it has a smaller 
	level number and thus is closer to the root.
\end{defi}

\noindent
Figure~\ref{fig:disk}(d) illustrates the ringed tree $RT(6)$.
Note that the diameter of the ringed tree $RT(k)$ is $\Theta(\log n)$, where 
$n=2^k-1$ is the number of vertices in $RT(k)$, and we will use $RT(\infty)$ 
to denote the infinite ringed tree when $k$ in $RT(k)$ goes to infinity.
Thus, a ringed tree may be thought of as a soft version of a binary tree; and 
to some extent, one can view a ringed tree as an idealized 
	picture reflecting the hierarchical
	structure in real networks coupled with local neighborhood connections,
	such as Internet autonomous system (AS) networks, which has both 
	a hierarchical structure of different level of AS'es, and peer connections
	based on geographical proximity. 






\subsection{Results and their implications}
\label{sec:rtresult}

A visual comparison of the ringed tree of Figure~\ref{fig:disk}(d) with
the tessellation of Poincar\'{e} disk (Figure~\ref{fig:disk}(b)) suggests 
that the ringed tree can been seen as an approximate tessellation or 
coarsening of the Poincar\'{e} disk.
Our first result in this section makes this precise; in particular, 
we show that the 
infinite ringed tree and the Poincar\'{e} disk are quasi-isometric.
\begin{thm} \label{thm:isometry}
The infinite ringed tree $RT(\infty)$ and the Poincar\'{e} disk
	are quasi-isometric.
\end{thm}

\noindent
Thus, by Proposition~\ref{prop:quasi}, we immediately have the following result.

\begin{coro}\label{cor:ringed-tree-constant}
There exists a constant $c$ 
s.t.,
for all $k$,
	ringed tree $RT(k)$ is $c$-hyperbolic.
\end{coro}

\noindent
Alternatively, we also provide a direct proof of this corollary
	(Section~\ref{sec:directrthyperbolicity}) 
to show that the ringed tree $RT(k)$ is Rips $5$-hyperbolic,
	and Gromov's $40$-hyperbolic
	in terms of the four point condition.
Our direct analysis also provides important properties of ringed trees that
	are used by later analyses.


Next, we address the question of whether long-range edges added at each 
level of the ring maintains or destroys the hyperbolicity of the base graph.
Given two vertices $u$ and $v$ at some level $t$ of the ringed tree, 
	we define the ring distance between $u$ and $v$, denoted $d_R(u,v)$,
	to be the length of the shorter path connecting $u$ and $v$ purely
	through the ring edges at the level $t$.
Given any function $f$ from positive integers to positive integers, 
	let $RT(k,f)$ denote the class of graphs constructed by adding
	long-range edges on the ringed tree $RT(k)$, such that for 
	each long-range edge $(u,v)$ connecting vertices $u$ and $v$ at the
	same level, $d_R(u,v) \le f(n)$, where $n=2^{k}-1$ is the number of
	vertices in the ringed tree $RT(k)$.
Since long-range edges do not reduce distances from root to any other
	vertices, the diameter of any graph in $RT(k,f)$ is still 
	$\Theta(\log n)$.
Define $\delta(RT(k,f)) = \max_{G\in RT(k,f)} \delta(G)$.

Our second result 
(used in the proof of the first part of our next result, but explicitly 
stated here since it is also of independent interest) 
is the following.

\begin{thm}\label{thm:random-ringed-tree-delta}
$\delta(RT(k,f)) = O(\log f(n))$, for any positive function $f$ and positive 
	integer $k$, where $n=2^{k}-1$ is the number of
	vertices in the ringed tree $RT(k)$.
\end{thm}

\noindent
This result 
indicates that if the long-range edges added do not
	span far-away vertices, then the graph should have good hyperbolicity.
In particular, if we take $f(n)=\log n$, then the theorem implies that
	the class $RT(k,f)$ is logarithmically hyperbolic.
The theorem covers all (deterministic) graphs in the class $RT(k,f)$.
We can extend it to random graphs, such that if we can show that 
	with high probability the random graph is in the class $RT(k,f)$, then
	we know that the hyperbolic $\delta$ of the random graph is
	$O(\log f(n))$ with high probability.
The first result in the next theorem is proven via this approach.

Next, we consider adding random edges between two vertices at the outermost
level, i.e., level $k-1$, such that the probability connecting two vertices 
	$u$ and $v$ is
	determined by a function $g(u,v)$.
Let $V_{k-1}$ denote the set of vertices at level $k-1$, i.e., the
	leaves of the original binary tree.
Given a real-valued positive function $g(u,v)$, 
	let $RRT(k,g)$ denote a random graph constructed as follows.
We start with the ringed tree $RT(k)$, and then for each 
	vertex $v\in V_{k-1}$, we add one long-range edge to a vertex 
	$u$ with
	probability proportional to $g(u,v)$, that is,
	with probability $g(u,v)\rho_v^{-1}$ where 
	$\rho_v = \sum_{u\in V_{k-1}} g(u,v)$.

We study three families of functions $g$, each of which has the characteristic
	that vertices closer to one another (by some measure) are more likely
	to be connected by a long-range edge.
The first two families use the ring distance $d_R(u,v)$ as the closeness
	measure.
In particular, the first family uses an exponential decay function
	$g_1(u,v) = e^{-\alpha d_R(u,v)}$.
The second family uses a power-law decay function
	$g_2(u,v) = d_R(u,v)^{-\alpha}$, where $\alpha > 0$.
The third family uses the height of the lowest common ancestor
	of $u$ and $v$, denoted as $h(u,v)$, as the closeness measure,
	and the function is $g_3=2^{-\alpha h(u,v)}$.
Note that this last probability function matches the function used by 
Kleinberg in a small-world model based on the tree structure~\cite{Kle01}.
Moreover, although $g_3$ and $g_2$ are similar, 
	in the ringed tree they are not the same, 
	since for two leaf nodes $u$ and $v$, $d_R(u,v)$ may not be the same as $2^{h(u,v)}$.
For example, let $u$ be the rightmost leaf of the left subtree of the root 
	(i.e. $u$ is represented as the string $01\ldots 1$) and $v$ be the leftmost leaf
	of the right subtree of the root (i.e. $v$ is represented as the string $10\ldots 0$),
	then $d_R(u,v)=1$ while $h(u,v)= \Theta(\log n)$.
The following theorem summarizes the hyperbolicity behavior of these three 
families of random ringed trees.

\begin{thm} \label{thm:randomRT}
Considering the follow families of functions (with $u$ and $v$
	as the variables of the function) for random ringed trees
	$RRT(k,g)$, for any positive integer $k$ and
	positive real number $\alpha$, 
	with probability $1-o(1)$ (when $n$ tends to infinity), 
	we have
\begin{enumerate}
\item
$\delta(RRT(k,e^{-\alpha d_R(u,v)}))=O(\log \log n)$;
\item
$\delta(RRT(k,d_R(u,v)^{-\alpha}))=\Theta(\log n)$;
\item
$\delta(RRT(k,2^{-\alpha h(u,v)}))=\Theta(\log n)$;
\end{enumerate}
where $n=2^{k}-1$ is the number of
	vertices in the ringed tree $RT(k)$.
\end{thm}


\noindent
This theorem states that,
when the random long-range edges are selected
	using exponential decay function based on the ring distance
	measure, the resulting graph
	is logarithmically hyperbolic, i.e., the constant hyperbolicity of
the original base graph is degraded only slightly;
but when a power-law decay function
	based on the ring distance measure or an exponential decay
	function based on common ancestor measure
	is used, then hyperbolicity is destroyed and the resulting graph is 
not hyperbolic.
One may notice that the function form in (1) and (3) above is similar
	but the result is different.
This is because with height $h(u,v)$ the subtree covers actually 
	$\Theta(2^{h(u,v)})$ leaves, and thus (3) is naturally closer to
	the power-law function of (2).
Intuitively, when it is more likely for a long-range edge 
	to connect two far-away vertices, such an edge 
	creates a shortcut for many internal tree nodes so that many shortest
	paths will go through this shortcut instead of traversing through
	tree nodes. (In Internet routing this is referred to as
	{\em valley routes}).

Finally, as a comparison, we also study the hyperbolicity of
	random binary trees $RBT(k,g)$, which are the same as
	random ringed trees $RRT(k,g)$ except that we remove all ring
	edges.

\begin{thm} \label{thm:randomRBT}
Considering the follow families of functions (with $u$ and $v$
	as the variables of the function) for random 
	binary trees $RBT(k,g)$, for any positive integer $k$ and
	positive real number $\alpha$, 
	with probability $1-o(1)$ (when $n$ tends to infinity), 
	we have
\[
\delta(RBT(k,e^{-\alpha d_R(u,v)}))=\delta(RBT(k,d_R(u,v)^{-\alpha}))=
	\delta(RBT(k,2^{-\alpha h(u,v)}))=\Theta(\log n),
\]
where $n=2^{k}-1$ is the number of
	vertices in the binary tree $RBT(k,g)$.
\end{thm}

\noindent
Thus, in this case, the original hyperbolicity of the base graph 
($\delta=0$ for the binary tree) is destroyed.
Comparing with Theorem~\ref{thm:randomRT}, 
	our results above suggest that the ``softening'' of the hyperbolicity 
provided by the rings is essential in maintaining
	good hyperbolicity: with rings, random ringed trees with
	exponential decay function (depending on the ringed distance) are logarithmically hyperbolic, but
	without the rings 
	the resulting graphs are not hyperbolic.

\subsection{Outline of the analysis} 
\label{sec:rtoutline}

In this subsection, we provide a summary of the proof of the four theorems 
in Section~\ref{sec:rtresult}.
For Theorem~\ref{thm:isometry}, we provide an embedding of the ringed tree
	to the Poincar\'{e} disk, intuitively similar to the picture we show
	in Figure~\ref{fig:disk}(d), and prove that it is a quasi-isometry.

For the analysis of $\delta$-hyperbolicity, we apply the Rips condition,
	which is equivalent to the Gromov's four point condition
	up to a constant factor. 
For any two vertices $u$ and $v$ on the ringed tree $RT(k)$, we define
	the canonical geodesic $\langle u, v \rangle$ to be the geodesic from
	$u$ to $v$ such that the geodesic always goes up first, then follows ring
	edges, and then goes down (any of these segments may be omitted).
We show that the canonical geodesic $\langle u, v \rangle$ and any other
	geodesic $[u,v]$ are within distance $1$ of each other, and 
	any triangle $\Delta(u,v,w)$ formed by three canonical geodesics
	$\langle u, v \rangle$, $\langle u, w \rangle$, and $\langle v, w \rangle$
	(called {\em canonical triangle}) are $3$-slim.
This immediately implies that any geodesic triangles in $RT(k)$ is
	$5$-slim, which is a direct proof that ringed trees are constantly
	hyperbolic.

For Theorem~\ref{thm:random-ringed-tree-delta}, we inductively prove that
	any geodesic $[u,v]$ in $RT(k,f)$ is within $O(\log f(n))$ distance
	from the canonical geodesic $\langle u,v\rangle$, and vice versa.
Together with the result that any canonical triangle is $3$-slim, it follows
	know that any geodesic triangle is $O(\log f(n))$-slim.
For Theorem~\ref{thm:randomRT}, Part~(1), we show that with high probability
	the long-range edges only connect vertices within ring distance
	$O(\log n)$, and then we can apply 
	Theorem~\ref{thm:random-ringed-tree-delta} to achieve the
	$O(\log\log n)$ bound.
For Theorem~\ref{thm:randomRT}, Part~(2), the key is to show that 
	(a) with high probability some long-range link connects two vertices
	at ring distance $\Theta(n^c)$ for some constant $c$; and 
	(b) if such a long-range edge $(u,v)$ exists, then we consider the
	geodesic triangle $\Delta(u,v,r)$ where $r$ is a point with lowest layer number on the canonical geodesic between $u,v$, and show that
	the middle point of $[r,u]$ is $\Theta(\log n)$ away from the union
	of $[r,v]$ and $[v,u]$.
For Theorem~\ref{thm:randomRT}, Part~(3), we first show that with high probability some pair of vertices $u,v$ with $h(u,v)\geq c \log_2 n$ for some constant $c>0$, and then we observe that, in such configuration, the ring distance $d_R(u,v)$ has high probability to be $\Omega(n^{c/2})$, and results follows exactly the same analysis in the previous part.

For Theorem~\ref{thm:randomRBT}, part (2) and (3) follow a similar
	strategy as those of Theorem~\ref{thm:randomRT}. 
For part (1), we know that two ``would-be'' ring neighbors $u$ and $v$
	have constant probability of having a long-range connection.
However, since we do not have ring edges, the alternative path between $u$
	and $v$ through the tree may be $\Theta(\log n)$ in length.
We show that there are at least $\Omega(\sqrt{n})$ such pairs, so with
	high probability at least
	one pair is connected, generating a bad $\delta$ of $\Omega(\log n)$.


\subsection{Detailed analysis on ringed trees} 
\label{sec:rtproof}

\subsubsection{Properties of ringed tree}

We start by some properties of ringed tree, which will be repeatedly used in the following analysis on ringed tree related graphs and which may be of independent interest.

We define the \emph{ring distance} $d_R(u,v)$ of $u$ and $v$ 
	on the same level to be their distance on the ring. 
Ringed trees have the following fundamental property.

\begin{lem}\label{lem:ringed-dist-growth}
Let $u$ and $v$ be two vertices on the same level,  
	and $u'$ and $v'$ be their parents respectively.
We have $d_R(u',v') \leq (d_R(u,v)+1)/2$.
\end{lem}

\begin{proof}
On the ring, there are $d_R(u,v)+1$ vertices on segment between $u$ and $v$, 
	which belong to at most $(d_R(u,v)+1)/2+1$ parents, which
	correspond to at most $(d_R(u,v)+1)/2+1$ vertices on the ring segment 
	between $u'$ and $v'$. This concludes the proof.
\end{proof}

For a geodesic $[u,v]$ on the ringed tree $RT(k)$, 
	we call its \emph{level sequence} the sequence of levels it passes by
	from $u$ to $v$. 
Lemma \ref{lem:ringed-dist-growth} implies that the level sequence 
	of any geodesic must be reversed unimodal: it first decreases, 
	and then increases (but the increasing or decreasing segment
	may be omitted).
The following lemma further characterizes geodesics in ringed-trees.

\begin{lem}\label{lem:ringed-tree-flipping}
Let $u,v$ be two vertices, $u$ be on level $\ell$, and $u'$ be
	the parent of $u$ at level $\ell-1$. 
Suppose $[u,v]$ intersects level $\ell-1$, and let $t$ be the intersection closest to $u$. Then $d(u,t) \leq 2$, the segment $[u,t]$ of $[u,v]$ and 
	$\{ u,u'\}$ are within distance $1$ to each other.
\end{lem}

\begin{proof}
Let $t'$ be the node just before $t$ to $u$ on $[t,u]$. $d(t,u) \leq 1 + d_R(t,u') \leq 1 + (d_R(t',u)+1)/2$ by Lemma \ref{lem:ringed-dist-growth}. As $t$ is the closest node on level $\ell-1$ on the geodesic to $u$, $d(t,u) \geq 1 + d_R(t',u)$. 
We get $d(t',u) \le 1$ and $d(t,u) \leq 2$ by combining these two inequalities. 
The segment $[u,t]$ of $[u,v]$ and $\{ u,u'\}$ are within distance $1$ to 
	each other since $d(t,u') \leq 1$.
\end{proof}




For two vertices $u$ and $v$, we define the {\em canonical geodesic} $\langle u,v \rangle$ in a recursive fashion.
\begin{enumerate}
\item For $u,v$ on the same level and $d_R(u,v) \leq 3$, $\langle u,v \rangle$ is the path on the ring from $u$ to $v$.
\item For $u,v$ on the same level but $d_R(u,v) > 3$, let $u', v'$ be parents of $u,v$ respectively, then $\langle u,v \rangle = 
	[u,u'] \cup \langle u', v' \rangle \cup [v',v]$.
\item 
For $u,v$ on different levels, supposing $u$ on upper level, let $v'$ be parent of $v$, then $\langle u,v \rangle = \langle u, v' \rangle \cup 
	[v',v]$.
\end{enumerate}
This is a well-founded definition. At each level of recursion, either the difference of levels of nodes decreases, or in the case of nodes on the same level, ring distance decreases by Lemma \ref{lem:ringed-dist-growth}, until we reach the base case, where $d_R(u,v) \leq 3$.

We prove now that canonical geodesics are really geodesics.

\begin{lem}\label{lem:canonical-geodesics-well-defined}
For any $u,v$, $\langle u,v \rangle$ is a geodesic between $u,v$.
\end{lem}

\begin{proof}
For $u,v$ on the same level and $d_R(u,v) \leq 3$, we can check that $\langle u,v \rangle$ is a geodesic between $u,v$.

For $u,v$ on the same level but $d_R(u,v) > 3$, let $u', v'$ be parents of $u,v$ respectively. Let $t, s$ be the closest node to $u,v$ to be in an upper level on $[u,v]$ respectively. If $t \neq u'$, then $d(u,t) \geq 2$ because the only way to go up one level in one step is to go to the parent. By Lemma \ref{lem:ringed-tree-flipping}, $[u,t]$ and $\{ u,u' \}$ are within distance $1$ to each other. Therefore $d(u',t)=1$, as $d(u,t) \geq 2$, and $\{u, u' \} \cup [t,v]$ is also a geodesic from $u$ to $v$. This is also correct for $t=u'$. The same can be proved for $s$. By combining, we have that $\{u, u' \} \cup [t,s] \cup \{ v', v \}$ is also a geodesic between $u,v$ for any geodesics $[t,s]$. If we pick the canonical geodesic $\langle t, s \rangle$, in any cases, this will be the canonical geodesic $\langle u, v \rangle$. Therefore, $\langle u, v \rangle$ is a geodesic.

For $u,v$ on different levels, the induction is essentially the same as in the previous case, but we only need to reason on $t$ only on the side of $u$.

This concludes our induction.
\end{proof}

\subsubsection{Proof of Theorem~\ref{thm:isometry}: 
	Quasi-isometry from infinite ringed tree to 
	the Poincar\'{e} disk}

In this subsection, we will exhibit and prove a quasi-isometry from the
	infinite ringed tree $RT(\infty)$ to the Poincar\'{e} disk. 
We denote its distance $d_{RT}$. We denote $(D,d_{P})$ the Poincar\'{e} disk, 
	where $D$ is the open disk of radius $1$ on the complex plane.

Here is a brief summary of our approach here. First we propose a candidate of quasi-isometry, then all possible cases of images of two points of ringed tree are divided into four categories, each of which
	is separately analyzed. We then proceed with an analysis on the metric of ringed tree and show that the candidate of quasi-isometry is effective, thus ringed tree and the Poincar\'{e} disk are quasi-isometric.

The following inequalities are used in the following.
\[ \ln(x) \leq \cosh^{-1}(x) \leq ln(2x) \quad \mathrm{for} \quad x \geq 1 \]
\[ \frac{2}{\pi}x \leq \sin(x) \leq x \quad \mathrm{for} \quad 0 \leq x \leq \frac{\pi}{2} \]
\[ 1 - \frac{x}{2} \leq \sqrt{1-x} \leq 1 - \frac{3x}{5} \quad \mathrm{for} \quad 0 \leq x \leq \frac{1}{2} \]

To state the promised quasi-isometry, we give coordinates to nodes in $RT(\infty)$. We know that we can number nodes with binary strings, which can be regarded as a number. For a node on the $k$-th level and numbered by $m$, its coordinates are $(k,m)$, with $0 \leq m \leq 2^{k}-1$. The root is level $0$.

\begin{defi}
Let the following mapping be the candidate of quasi-isometry :
\[ f: RT(\infty) \to D, (k,m) \mapsto \sqrt{1-2^{-k}} e^{2i\pi \frac{m}{2^{k}}} \]
For $0 \leq k \leq \ell$, $m < 2^{\ell-1}$, we define $D(k,\ell,m)=d_{P}(f(k,0),f(\ell,m))$. This is a distance in the Poincar\'{e} disk. Following is its full expression in $k$ and $m$.
\begin{eqnarray*}
D(k ,\ell ,m) & = & \cosh^{-1}(1+2\frac{\Vert \sqrt{1-2^{-k}} - \sqrt{1-2^{-\ell}} e^{2i\pi \frac{m}{2^{\ell}}} \Vert^{2}}{2^{-k-\ell}}) \\
& = & \cosh^{-1}(1 + 2(\sqrt{2^{\ell}(2^{k}-1)} - \sqrt{2^{k}(2^{\ell}-1)})^2 + 8\sqrt{2^{\ell}(2^{\ell}-1)2^{k}(2^{k}-1)}\sin^{2}(\pi \frac{m}{2^{\ell}}))
\end{eqnarray*}
We also define $D'(k,\ell,m)=d_{RT}(f(k,0),f(\ell,m))$, with $0 \leq k \leq \ell$, $m < 2^{\ell-1}$. This is a distance in ringed tree.
\end{defi}

We will now try to bound $D(k,\ell,m)$ with the following lemma.

\begin{lem}\label{lem:quasi-iso-bound}
We have the following bounds on $D(k, \ell, m)$.
\begin{enumerate}
\item For $k=0$,
\[ \frac{\ln(2)}{2}\ell + \frac{\ln(2)}{2} \leq D(0, \ell, m) \leq \frac{\ln(2)}{2}\ell + \ln(6). \]
\item For $0 < k=\ell, m > 0$,
\[ \ln(2)(4 + 2\lfloor \log_2 m \rfloor) \leq D(k,k,m) \leq \ln(2)(4 + 2\lfloor \log_2 m \rfloor) + \ln(\frac{5}{4\pi^{2}}) \mathrm{for} (k \geq 1). \]
\item For $m=0, 1 \leq k < \ell$, 
\[ \ln(2)(\ell - k) - \ln(50) \leq D(k, \ell, 0) \leq \ln(2)(\ell - k). \]
\item For $0 < k < \ell$ and $0 < m < 2^{\ell-k}$,
\[ \ln(2)(\ell - k) - \ln(100) \leq D(k,\ell,m) \leq \ln(2)(\ell - k) + \ln(66 \pi^{2}). \]
\item For $0 < k < \ell$ and $2^{\ell-k} \leq m < 2^{\ell - 1}$,
\[ \ln(2)(k - \ell + 2\lfloor \log_2 m \rfloor + 4) \leq D(k,\ell,m) \leq \ln(2)(k - \ell + 2\lfloor \log_2 m \rfloor + 6) + \ln(\pi^{2} + 1). \]
\end{enumerate}
\end{lem}

\begin{proof}
\begin{enumerate}
\item Case $k=0$
\[ D(0, \ell, m) = \cosh^{-1}(1 + 2\sqrt{2^{\ell}-1}) \]
If $\ell=0$, $D(0,0,m)=0$. For $\ell \geq 1$, we have
\[ D(0, \ell, m) \geq \cosh^{-1}(2\sqrt{2^{\ell-1}}) \geq \frac{\ln(2)}{2}\ell + \frac{\ln(2)}{2} \]
and
\[ D(0, \ell, m) \leq \cosh^{-1}(3\sqrt{2^{\ell}}) \leq \frac{\ln(2)}{2}\ell + \ln(6). \]
\item Case $0<k=\ell, m>0$
\[ D(k, k, m) = \cosh^{-1}(1 + 2^{k+3}(2^{k}-1)\sin^{2}(\frac{m}{2^{k}}\pi)) \]
Let $a = \lfloor \log_2 m \rfloor$, we have
\[ D(k, k, m) \geq \cosh^{-1}(2^{2k+2} \sin^{2}(2^{a-k}\pi)) \geq \cosh^{-1}(2^{2k+2}(2^{a+1-k})^{2}) \geq \ln(2)(4 + 2a) \]
and
\begin{eqnarray*}
D(k, k, m)&\leq &\cosh^{-1}(\frac{5}{4} 2^{2k+2} \sin^{2}(2^{a+1-k}\pi)) \leq \cosh^{-1}(\frac{5}{4} 2^{2k+2} (2^{a+1-k}\pi)^{2}) \\
&\leq &\ln(2)(4 + 2a) + \ln(\frac{5}{4\pi^{2}})
\end{eqnarray*}
\item Case $0 < k < \ell, m=0$
\begin{eqnarray*}
D(k, \ell, 0) &=& \cosh^{-1}(1 + 2(\sqrt{2^{\ell}(2^{k}-1)} - \sqrt{2^{k}(2^{\ell}-1)})^{2}) \\
& = &\cosh^{-1}(1 + 2^{k + \ell + 1}(\sqrt{1 - 2^{-k}} - \sqrt{1 - 2^{-\ell}})^{2})
\end{eqnarray*}
As $1 \leq k < \ell$, $\sqrt{1 - 2^{-\ell}} - \sqrt{1 - 2^{-k}} > 0$, and we have
\[ \sqrt{1 - 2^{-\ell}} - \sqrt{1 - 2^{-k}} \geq 1 - 2^{-\ell-1} - 1 + \frac{3}{5}2^{-k} \geq \frac{1}{5}2^{-k-1} \]
and
\[ \sqrt{1 - 2^{-\ell}} - \sqrt{1 - 2^{-k}} \leq 1 - \frac{3}{5}2^{-\ell} - 1 + 2^{-k-1} \leq 2^{-k-1} \]
Therefore, we have $ \ln(2)(\ell - k) - \ln(50) \leq D(k, \ell, 0) \leq \ln(2)(\ell - k) $.
\item Case $0 < k < \ell, 0 < m < 2^{\ell-1}$

We now deal with the general case with $m > 0$ and $0 < k < \ell$.
\[ D(k, \ell, m) = \cosh^{-1}(\cosh(D(k, \ell, 0)) + 8\sqrt{2^{\ell}(2^{\ell}-1)2^{k}(2^{k}-1)}\sin^{2}(\pi \frac{m}{2^{\ell}})) \]
We always note $a = \lfloor \log_2 m \rfloor$, and we have
\[ \sqrt{2^{\ell}(2^{\ell}-1)2^{k}(2^{k}-1)}\sin^{2}(\pi \frac{m}{2^{\ell}}) \geq 2^{\ell + k - 1} 2^{2a - 2\ell +2} = 2^{k - \ell + 2a + 1} \]
and
\[ \sqrt{2^{\ell}(2^{\ell}-1)2^{k}(2^{k}-1)}\sin^{2}(\pi \frac{m}{2^{\ell}}) \leq 2^{\ell + k} 2^{2a + 2 - 2\ell} \pi^{2} \leq 2^{k - \ell + 2a + 2} \pi^{2} \]
The previous bound on $D(k, \ell, 0)$ transforms into the following by applying $\cosh$.
\[ \frac{2^{\ell - k}}{100} \leq \cosh(D(k, \ell, 0)) \leq 2^{\ell - k} \]
Suitable substitution of $\cosh(D(k, \ell, 0))$ gives
\[ \cosh^{-1}(\frac{2^{\ell - k}}{100} + 2^{k - \ell + 2a + 4}) \leq D(k, \ell, m) \leq \cosh^{-1}(2^{\ell - k} + 2^{k - \ell + 2a + 5} \pi^{2}). \]
For $0 \leq a \leq \ell - k$, therefore $0 < m < 2^{\ell-k}$,
\[ D(k, \ell, m) \geq \cosh^{-1}(\frac{2^{\ell - k}}{100} + 2^{k - \ell + 2a + 4}) \geq \ln(\frac{2^{\ell - k}}{100}) = \ln(2)(\ell - k) - \ln(100) \]
and
\begin{eqnarray*}
D(k, \ell, m) & \leq & \cosh^{-1}(2^{\ell - k} + 2^{k - \ell + 2a + 5} \pi^{2}) \\
& \leq & \ln(2^{\ell - k} 33 \pi^{2}) + ln(2) = \ln(2)(\ell - k + 1) + \ln(33 \pi^{2}).
\end{eqnarray*}
For $\ell - k + 1 \leq a < m - 1$, therefore $2^{\ell-k} \leq m < 2^{\ell - 1}$,
\[ D(k, \ell, m) \geq \cosh^{-1}(\frac{2^{\ell - k}}{100} + 2^{k - \ell + 2a + 4}) \geq \ln(2^{k - \ell + 2a + 4}) = \ln(2)(k - \ell + 2a + 4) \]
and
\begin{eqnarray*}
D(k, \ell, m) & \leq & \cosh^{-1}(2^{\ell - k} + 2^{k - \ell + 2a + 5} \pi^{2}) \\
& \leq & \ln((\pi^{2} + 1)2^{k - \ell + 2a + 5}) + ln(2) = \ln(2)(k - \ell + 2a + 6) + \ln(\pi^{2} + 1).
\end{eqnarray*}
\end{enumerate}
\end{proof}

We will now try to relate $D(k,\ell,m)$ and $D'(k,\ell,m)$ in the following lemma.

\begin{lem}\label{lem:bound-of-two-distances}
For any $0 \leq k \leq \ell, 0 \leq m < 2^{\ell - 1}$,
\[ \frac{\ln(2)}{2}D'(k,\ell,m) - \ln(200) \leq D(k,\ell,m) \leq \ln(2)D'(k,\ell,m) + \ln(66 \pi^{2}) \]
\end{lem}

\begin{proof}
Consider the canonical geodesic. If $k = \ell$ and $m=1$, it is an edge from $(k,0)$ to $(k,1)$. In other cases, it goes up first from $(\ell,m)$ to a certain ancestor, then makes $2$ or $3$ moves on the ring, and finished by going straight down to $(k,0)$.

In the first case, $D'(k,k,1)=1$. In the second case, as the ancestor of any node $(k,m)$ is $(k-1, \lfloor \frac{m}{2} \rfloor)$, by the form of canonical geodesics, we should first go up $\ell - k$ steps to reach level $k$. If $m \leq 2^{\ell - k + 1}$, it reaches $(k,0)$ by at most an extra step. In this case, $\ell - k \leq D'(k,\ell,m) \leq \ell - k + 1$. If $m > 2^{\ell - k + 1}$, we go up $\lfloor \log_2 m \rfloor - 1$ steps from $(\ell,m)$ to reach an ancestor numbered $2$ or $3$, then go $2$ or $3$ steps on the ring to the node numbered $0$ on the same level, and finish by going down to $(k,0)$. Thus we have $2\lfloor \log_2 m \rfloor + k - \ell \leq D'(k,\ell,m) 2\lfloor \log_2 m \rfloor + k - \ell + 1$ in this case.

We conclude by comparing to bounds in Lemma~\ref{lem:quasi-iso-bound}.
\end{proof}
We want to bound distance between any two points in $RT(\infty)$ with the following lemma.

\begin{lem} \label{lem:dist-and-ring-dist}
For $u_1, v_1, u_2, v_2 \in RT(\infty)$ with $u_1,v_1,u_2,v_2$ on the same level and $d_R(u_1,v_1)=d_R(u_2,v_2)$, we have $|d(u_1,v_1)-d(u_2,v_2)| \leq 3$.
\end{lem}

\begin{proof}
Let $A=d_R(u_1,v_1)$. It is clearly correct when $A \leq 3$ by the structure of canonical geodesics.

For $A \leq 3$, when we consider the canonical geodesics, we take successive ancestors of $u_1$ and $v_1$ until their distance is $2$ or $3$. This takes at least $\lfloor \log_2 A \rfloor - 1$ generations, but at most $\lfloor \log_2 (A-1) \rfloor$. These bounds differ by at most $1$. Platforms differ by at most $1$, therefore $|d(u_1,v_1)-d(u_2,v_2)| \leq 3$.
\end{proof}

\begin{proof}[Proof of Theorem~\ref{thm:isometry}]
By symmetry, Lemma \ref{lem:dist-and-ring-dist} and Lemma~\ref{lem:bound-of-two-distances}, for any $u,v \in Rt(\infty)$, we have:
\[ \frac{\ln(2)}{2}d_{RT}(u,v) - \ln(200) \leq d_{P}(f(u),f(v)) \leq \ln(2)d_{RT}(u,v) + \ln(66 \pi^{2}) \]

There is only one thing left to prove quasi-isometry. We will now prove that for some constant $\epsilon$, $B(f(RT(\infty)),\epsilon)$ covers the Poincar\'{e} disk. Images of each level are all on concentric circles, and the difference of radius between successive levels can be bounded by $\ln(6)$. Distance between images of neighboring nodes on the same level can be bounded by $\ln(16)$. For any point in the Poincar\'{e} disk, its distance to the nearest image of nodes is bounded by $\ln(96)$, by first moving straight away from $0$ until reaching a concentric circle of images, then take the shortest path to reach image of a certain node.

This proves $f$ to be a $(\frac{2}{\ln(2)},\ln(66 \pi^{2}))$-quasi-isometry from $RT(\infty)$ to the Poincar\'{e} disk, thus $RT(\infty)$ and the Poincar\'{e} disk are quasi-isometric. Constants we found here are not tight.
\end{proof}

\subsubsection{A direct proof of Corollary~\ref{cor:ringed-tree-constant}}
\label{sec:directrthyperbolicity}

We begin with a lemma about the distance between a general geodesic and the corresponding canonical geodesic.

\begin{lem}\label{lem:rectified-geodesic}
For any geodesic $[u,v]$, $[u,v]$ and $\langle u,v \rangle$ are within distance $1$ to each other.
\end{lem}

\begin{proof}
We perform an induction on the structure of $\langle u,v \rangle$.

For $u,v$ on the same level and $d_R(u,v) \leq 3$, we can check that $[u,v]$ and $\langle u,v \rangle$ are within distance $1$ to each other.

For $u,v$ on the same level but $d_R(u,v) > 3$, let $u', v'$ be parents of $u,v$ respectively. We have $\langle u,v \rangle = 
	\langle u, u'\rangle \cup \langle u', v' \rangle \cup 
	\langle v',v\rangle$ by construction. 
By Lemma \ref{lem:ringed-tree-flipping}, the parts of $[u,v]$ that are
	on the same level with $u,v$ verify the condition already. We only need to deal with the part on levels with lower numbering.

Let $t, s$ be the closest node to $u,v$ that are in an upper level on 
	$[u,v]$ respectively. 
If $t \neq u'$, then $d(u,t) \geq 2$ because the only way to go up one level in one step is to go to the parent. By Lemma \ref{lem:ringed-tree-flipping}, $[u,t]$ and $\langle u,u' \rangle$ are within distance $1$ to each other. Therefore $d(u',t)=1$, as $d(u,t) \geq 2$, 
	$\langle u, u' \rangle \cup \langle u', t \rangle \cup[t,v]$ is also a geodesic from $u$ to $v$. This is also correct for $t=u'$. The same can be proven for $s$. 
By combining the above, we have that $\langle u, u' \rangle \cup 
	\langle u', t \rangle
	\cup [t,s] \cup \langle s, v'\rangle \cup \langle v', v \rangle$ is also a geodesic between $u,v$.
Thus the section $\langle u', t \rangle
	\cup [t,s] \cup \langle s, v'\rangle$ is also a geodesic from $u'$ to $v'$. By induction hypothesis, it is within distance $1$ to $\langle u', v' \rangle$. However, this geodesic contains the part of $[u,v]$ on levels with lower numbering than that of $u,v$, which is precisely $[t,s]$. We conclude that $[u,v]$ and $\langle u,v \rangle$ are within distance $1$ to each other.

For $u,v$ on different levels, the induction is essentially the same as in the previous case, but we only need to reason on $t$ on the side of $u$.

This concludes our induction.
\end{proof}

\begin{lem} \label{lem:smallcan}
Let $u$ and $v$ be two vertices at the same level $\ell$ such that 
	$\langle u,v\rangle$ stays at level $\ell$.
Then for any vertex $w$, $\langle u,w\rangle$ and 
	$\langle v, w\rangle$ are within distance $3$ of each other.
\end{lem}
\begin{proof}
Let $\ell_0$ and $\ell_0'$ be the highest levels reached by
	$\langle u,w\rangle$ and $\langle v, w\rangle$, respectively.
Without loss of generality, we assume that $\ell_0 \ge \ell_0'$.
Thus $\ell \ge \ell_0$.
We prove the lemma by an induction on $\ell$.

Consider the base case of $\ell = \ell_0$.
Let $w'$ be the ancestor of $w$ at level $\ell_0$.
Thus we know that both $\langle w', u\rangle$ and $\langle u,v\rangle$
	stay at level $\ell_0$.
By definition $d_R(w',u)\le 3$ and $d_R(u,v)\le 3$.
We can enumerate all the possible cases of $u,v,w'$ arrangement to check
	that $\langle w', v \rangle$ is at most one level higher than
	$\ell_0$ (i.e., $\ell_0' \ge \ell_0 -1 $), and $\langle w', v \rangle$ and
	$\langle w', u\rangle$ are within distance $3$ of each other.
Since $\langle u,w\rangle$ and $\langle v, w\rangle$ share the portion
	$\langle w, w'\rangle$, the lemma holds for the case of $\ell=\ell_0$.

For the induction step, consider $\ell > \ell_0$.
Let $u'$ and $v'$ be the parents of $u$ and $v$, respectively.
By Lemma~\ref{lem:ringed-dist-growth} and $d_R(u,v)\le 3$, we know that
	$d_R(u',v') \le 2$.
Then $\langle u',v' \rangle$ must stay at level $\ell-1$.
By induction hypothesis, $\langle u',w \rangle$ and $\langle v',w \rangle$
	are within distance $3$ of each other.
Since $\langle u,w \rangle = \langle u,u' \rangle \cup
	\langle u',w \rangle$, $\langle v,w \rangle = \langle v,v' \rangle \cup
	\langle v',w \rangle$, and $d_R(u,v)\le 3$, we know that the 
	lemma holds in this case.
\end{proof}

We define {\em canonical triangle} $\hat{\Delta}(u,v,w)$ to be 
	a geodesic triangle in which all sides are canonical geodesics. 

\begin{lem}\label{lem:canonical-geodesic-triangle-thin}
Any canonical triangle $ \hat{\Delta}(u,v,w)$ in $RT(k)$ is $3$-slim.
\end{lem}

\begin{proof}
Without loss of generality, let $u$ be the vertex at 
	the lowest level, which is $\ell$.
We prove the lemma by an induction on $\ell$.
The base case of $\ell = 0$ is trivial.

Consider the induction step with $\ell > 0$.
If neither $v$ nor $w$ is at level $\ell$, then 
	both $\langle u, v\rangle$ and $\langle u,w\rangle$ go through
	$u$'s parent $u'$.
By induction hypothesis, we know that $ \hat{\Delta}(u',v,w)$ is
	$3$-slim.
Adding $u$ in this case does not change the distance among the sides, and
	thus $ \hat{\Delta}(u,v,w)$ is also $3$-slim.
Suppose now that $v$ or $w$ or both are at level $\ell$.
In the first case, suppose that at least one pair, say $u$ and $v$, is such that
	$\langle u, v\rangle$ stays at level $\ell$.
By Lemma~\ref{lem:smallcan}, we know that $\langle u, w\rangle$ and
	$\langle v, w\rangle$ are within distance $3$ of each other.
Since $\langle u, v\rangle$ has length at most $3$, we know that
	$ \hat{\Delta}(u,v,w)$ is $3$-slim.
In the second case, suppose that all pairs at level $\ell$ have
	their canonical geodesics into level $\ell -1$.
For each vertex at level $\ell$, we take its parent, together with perhaps
	another vertex already within level $\ell-1$, we can apply the
	induction hypothesis and show that their canonical triangle is
	$3$-slim.
Since for every vertex at level $\ell$, its canonical geodesics
	to the other two vertices all go through its parent, 
	the vertices at level $\ell$ do not change the distance between
	any pair of sides of the canonical triangle.
Therefore, $ \hat{\Delta}(u,v,w)$ is $3$-slim.
\end{proof}



With Lemmas~\ref{lem:rectified-geodesic} 
	and~\ref{lem:canonical-geodesic-triangle-thin}, we can provide
	a direct proof of Corollary~\ref{cor:ringed-tree-constant}.
\begin{proof}[Direct proof of Corollary~\ref{cor:ringed-tree-constant}]
We use Rips condition here. 
For any $u,v,w$, 
By Lemma \ref{lem:rectified-geodesic}, we have
	$[u,v] \subseteq B(\langle u,v \rangle,1)$.
By Lemma \ref{lem:canonical-geodesic-triangle-thin}, we have
	$\langle u,v \rangle \subseteq B(\langle u,w \rangle \cup 
	\langle v,w \rangle , 3)$.
By Lemma \ref{lem:rectified-geodesic} again, we have
	$\langle u,w \rangle \cup \langle v,w \rangle \subseteq 
	B([u,w] \cup [v,w],1)$. 
Therefore, we conclude that $[u,v] \subseteq B([u,w] \cup [v,w],5)$, 
	and thus $\delta_{Rips}(RT(k))\le 5$.
Since $\delta(RT(k)) \le 8\delta_{Rips}(RT(k))$, we have
	$\delta(RT(k))\le 40$.
\end{proof}


\subsubsection{Proof of Theorem~\ref{thm:random-ringed-tree-delta}}

We now apply Rips condition to analyze the hyperbolicity of ringed trees
	with limited long-range edges $RT(k,f)$.
First we show that a geodesic in $RT(k,f)$ 
	cannot make too many hops at the same level of the ringed tree.
For any two vertices in a graph $G \in RT(k,f)$, let $[u,v]$
	denote any one of the geodesics between $u$ and $v$ in $G$, while
	$\langle u, v \rangle$ still denote the canonical 
	geodesic between $u$ and $v$ in
	the base ringed tree $RT(k)$.


\begin{lem}\label{lem:dist-estm-by-ring-dist}
For any $u,v$ on the same level in $RT(k)$ with $d(u,v) > 1$, $2\log_2 d_R(u,v) \leq d(u,v) \leq 2\log_2 (d_R(u,v)-1) + 2$.
\end{lem}

\begin{proof}
We perform induction on $d_R(u,v)$. If $1 < d_R(u,v) \leq 3$, we can check that it is correct. If $d_R(u,v) > 3$, let $u', v'$ be parents of $u,v$ respectively. By induction hypothesis, $2\log_2 d_R(u',v') \leq d(u',v') \leq 2\log_2 (d_R(u',v')-1) + 2$. 
By triangle inequality, $d(u,v) \le d(u',v') + 2 \leq 2\log_2 (d_R(u',v')-1) + 4$. But $d_R(u',v') - 1 \leq (d_R(u,v)-1)/2$ by Lemma \ref{lem:ringed-dist-growth}. And we have $d(u,v) \leq 2\log_2 (d_R(u',v')-1) + 4 \leq 2\log_2 (d_R(u,v)-1) + 2$. On the other hand, $d(u,v) \geq 2 + d(u',v') \geq 2\log_2 (2d_R(u',v'))$. As $d_R(u',v') \geq d_R(u,v)/2$, we have $d(u,v) \geq 2\log_2 d_R(u,v)$. Combining this two inequalities concludes the induction.
\end{proof}

\begin{coro}\label{coro:dist-estm-by-ring-dist}
For any $u,v$ on the same level in $RT(k)$, $d(u,v) \leq 2\log_2 d_R(u,v) + 2$.
\end{coro}

\begin{proof}
We check for the case $d(u,v)=1$ and it is satisfied. 
For $d(u,v)>1$, this results directly from $d(u,v) \leq 2\log_2 (d_R(u,v)-1) + 2 \leq 2\log_2 d_R(u,v) + 2$.
\end{proof}

\begin{lem}\label{lem:limited-jumps}
For a graph in the class $RT(k,f)$ and two vertices
	$u$ and $v$ at the same level $j$, 
	if $[u,v]$ never goes into vertices at level $i<j$, 
	then $d(u,v) \leq \max(32,4\log_2 f(n))-1$ and 
	$d_R(u,v) \leq f(n)\max(32,4\log_2 f(n))-1$.
\end{lem}
\begin{proof}
First, we will prove $d_R(u,v)/f(n) \leq d(u,v) \leq 2\log_2 d_R(u,v) + 2$.
The first inequality is because the shortest distance from $u$ to $v$
	without going into vertices in level $i < j$ in the base ringed
	tree $RT(k)$ is $d_R(u,v)$, and each long-range edge can jump at most
	$f(n)$ hops on the ring at any level.
The second inequality results directly from Corollary \ref{coro:dist-estm-by-ring-dist}.
Let $ x=d(u,v)$.
From the above two inequalities,  
	we have $x \leq 2\log_2 xf(n) + 2$.
Then, $f(n) \geq 2^{(x-2)/2}/x \geq 2^{(x+1)/4}$ for $x \geq 32$, 
	and thus $x=d(u,v) \leq \max(32,4\log_2 f(n))-1$. 
It follows that $d_R(u,v) \leq f(n)\max(32,4\log_2 f(n))-1$.
\end{proof}



\begin{lem} \label{lem:loglocal}
For a graph in the class $RT(k,f(n))$ and two
	vertices $u$ and $v$, $[u,v]$ and $\langle u,v\rangle$ 
	are within distance $2\max(32,4\log_2 f(n))$ to each other.
\end{lem}

\begin{proof}
The case of $u$ and $v$ are ancestor and descendant to each other
	on the tree are trivial.
Thus we consider the case that $u$ and $v$ are not 
	ancestor and descendant to each other.
Let $\ell_0^{\langle\rangle}$ and  $\ell_0^{[\,]}$ be the innermost level that 
	geodesics $\langle u,v\rangle$ and $[u,v]$ reach, respectively.
If $\ell_0^{[\,]} < \ell_0^{\langle\rangle}$, then
	$[u,v]$ uses at least two more tree edges than $\langle u,v\rangle$.
This means $\langle u,v\rangle$ uses at least two ring edges.
If $\langle u,v\rangle$ uses exactly two ring edges, 
	all edges in $[u,v]$ must be tree edges, which is impossible
	to be a geodesic given that $u$ and $v$ are not 
	ancestor and descendant to each other.
Thus, $\langle u,v\rangle$ uses exactly three ring edges.
Then $[u,v]$ uses exactly one ring or long-range edge.
Let $u'$ and $v'$ be the ancestors of $u$ and $v$ at level 
	$\ell_0^{\langle\rangle}$, respectively.
In this case, the only possible situation is : 
	(a) $\ell_0^{[\,]} = \ell_0^{\langle\rangle}-1$, and
	(b) the parents of $u'$ and $v'$ are connected either by a ring edge
	or a long-range edge, which is used by $[u,v]$.
Hence, $[u,v]$ and $\langle u,v\rangle$ share the tree edges
	from $u$ to $u'$ and $v'$ to $v$, and $\langle u,v\rangle$
	goes through ring edges from $u'$ to $v'$ while 
	$[u,v]$ goes through the edge connecting the parents of $u'$ and $v'$.
We thus have that $[u,v]$ and $\langle u,v\rangle$ are within
	distance $1$ of each other.
Therefore, from now on, we consider $\ell_0^{[\,]} \ge \ell_0^{\langle\rangle}$.

Let $\ell_u$ and $\ell_v$ be
	the levels of $u$ and $v$ respectively. 
Without loss of generality, We can suppose that $\ell_u \geq \ell_v \geq 
	\ell_0^{[\,]}$.

For any level $\ell$ both reachable by $[u,v]$ and $\langle u,v\rangle$,
	i.e., $\ell_0^{[\,]} \le \ell \le \ell_u$, 
	let $x_\ell$ be the first level-$\ell$ vertex on 
	geodesic $[u,v]$ starting from $u$, and let $y_\ell$ be the first 
	level-$\ell$ vertex
	on geodesic $\langle u,v\rangle$ starting from $u$.
We claim that $d_R(x_\ell,y_\ell) \le f(n)\max(32,4\log_2 f(n))$.


For level $\ell_u$, it is trivial. 
Suppose that our claim is correct for some level $\ell > \ell_0^{[\,]}$, 
	and we inductively
	prove the claim for level $\ell-1$.
By the induction hypothesis, 
	we know that $d_R(x_\ell,y_\ell) \leq f(n)\max(32,4\log_2 f(n))$. 
Let $x'_\ell$ be the level-$\ell$ vertex just before $x_{\ell-1}$ on $[u,v]$ 
	starting from $u$. 
By Lemma \ref{lem:limited-jumps} and the fact that the portion $[x_\ell,x'_\ell]$ 
	never goes to level $i < \ell$, we have
	$d_R(x_\ell,x'_\ell) \leq f(n)\max(32,4\log_2 f(n))-1$ . 
Therefore, $d_R(x'_\ell,y_\ell) \leq d_R(x_\ell,y_\ell) + d_R(x_\ell,x'_\ell) \leq 2f(n)\max(32,4\log_2 f(n))-1$.
Since $x_{\ell-1}$ and $y_{\ell-1}$ are 
	the parents of $x'_\ell$ and $y_\ell$ respectively,
	by Lemma \ref{lem:ringed-dist-growth} we have
	$d_R(x_{\ell-1},y_{\ell-1})$ $\leq $ $f(n)\max(32,4\log_2 f(n))$. 
Our claim holds for level $\ell - 1$. By induction, our claim stands.

We thus have $d(x_{\ell},y_{\ell}) \leq 2\log_2(d_R(x_{\ell},y_{\ell})) + 
	2 \leq \max(32,4\log_2 f(n))$, for all $\ell_0^{[\,]} \le \ell \le \ell_u$. 
For any vertex $x$ between $x_\ell$ and $x'_{\ell}$ on $[u,v]$ (note
	that $x$ may be at a level $\ell'\ge \ell$),
	by Lemma \ref{lem:limited-jumps} we have
	$d(x,y_\ell) \le d(x,x_{\ell}) + d(x_{\ell},y_{\ell}) \le
	d(x'_{\ell},x_{\ell}) + d(x_{\ell},y_{\ell}) \le 2\max(32,4\log_2 f(n))$.
Hence all such vertices $x$ are within distance $2\max(32,4\log_2 f(n))$
	from vertex $y_\ell$ in $\langle u,v\rangle$, and vice versa.

Similarly, we can define $z_\ell$ to be the the first level-$\ell$ vertex on 
	geodesic $[v,u]$ starting from $v$, and  $w_\ell$ to be the first 
	level-$\ell$ vertex
	on geodesic $\langle v,u\rangle$ starting from $v$, for
	all $\ell_0^{[\,]} \le \ell \le \ell_v$.
By a symmetric argument, we can show that $w_\ell$ are within 
	distance $2\max(32,4\log_2 f(n))$ from all vertices in 
	the segment of $[u,v]$ from $z_{\ell-1}$ to $z_{\ell}$ for 
	$\ell > \ell_0^{[\,]}$, and $d(z_{\ell_0^{[\,]}}, w_{\ell_0^{[\,]}}) \le
	\max(32,4\log_2 f(n))$.

The only portion left to argue is from $x_{\ell_0^{[\,]}}$ to $z_{\ell_0^{[\,]}}$
	in geodesic $[u,v]$, and from $y_{\ell_0^{[\,]}}$ to $w_{\ell_0^{[\,]}}$ in
	geodesic $\langle u,v\rangle$.
By Lemma \ref{lem:limited-jumps} and the definition of $\ell_0^{[\,]}$, 
	we know that 
	$d(x_{\ell_0^{[\,]}}, z_{\ell_0^{[\,]}}) \le
	\max(32,4\log_2 f(n))$.
Therefore, all vertices in the segment from $x_{\ell_0^{[\,]}}$ to $z_{\ell_0^{[\,]}}$
	in geodesic $[u,v]$ are within $2\max(32,4\log_2 f(n))$ to
	both $y_{\ell_0^{[\,]}}$ and $w_{\ell_0^{[\,]}}$.
Now, for any vertex $x$ in the segment from $y_{\ell_0^{[\,]}}$ to $w_{\ell_0^{[\,]}}$ in
	geodesic $\langle u,v\rangle$, we need to bound the distance from 
	$x$ to $[u,v]$.
Since $d_R(y_{\ell_0^{[\,]}},w_{\ell_0^{[\,]}}) \le
	d_R(y_{\ell_0^{[\,]}},x_{\ell_0^{[\,]}}) + d_R(x_{\ell_0^{[\,]}},z_{\ell_0^{[\,]}}) + 
	d_R(z_{\ell_0^{[\,]}},w_{\ell_0^{[\,]}}) \le 3 f(n)\max(32,4\log_2 f(n))$,
	we have $d(y_{\ell_0^{[\,]}},w_{\ell_0^{[\,]}})  \le
	2\log_2d_R(y_{\ell_0^{[\,]}},w_{\ell_0^{[\,]}}) + 2\le \max(32,4\log_2 f(n))$.
Thus, for any vertex $x$  in the segment from $y_{\ell_0^{[\,]}}$ to 
	$w_{\ell_0^{[\,]}}$ in geodesic $\langle u,v\rangle$, $x$ can 
	reach either $x_{\ell_0^{[\,]}}$ and $z_{\ell_0^{[\,]}}$ in
	at most $2\max(32,4\log_2 f(n))$
	hops.
Therefore, $x$ is within distance $2\max(32,4\log_2 f(n))$ from $[u,v]$.
\end{proof}

\begin{proof}[Proof of Theorem~\ref{thm:random-ringed-tree-delta}]
We use Rips condition here. For any $u,v,w$, 
	by Lemma~\ref{lem:loglocal}, we have
	$[u,v] \subseteq B(\langle u,v \rangle,$ $ 2\max(32,4\log_2 f(n)))$.
By Lemma \ref{lem:canonical-geodesic-triangle-thin}, 
	$\langle u,v \rangle \subseteq B(\langle u,w \rangle \cup 
	\langle v,w \rangle , 3)$.
Again by Lemma~\ref{lem:loglocal}, 	$\langle u,w \rangle \cup 
	\langle v,w \rangle \subseteq B([u,w] \cup [v,w],2\max(32,4\log_2 f(n)))$.
Combining these together, we have
	$[u,v] \subseteq B([u,w] \cup [v,w],4\max(32,4\log_2 f(n))+3)$.
Therefore, since $\delta$ and $\delta_{Rips}$ differ within a constant factor, we have $\delta(RT(k,f)) \le c \log f(n)$ for some constant $c$.
\end{proof}

\subsubsection{Proof of Theorem~\ref{thm:randomRT}}




We first analyze the $\delta$-hyperbolicity of $RRT(k,e^{-\alpha d_R(u,v)})$.

\begin{proof}[Proof of Theorem \ref{thm:randomRT}, part 1]
$e^{\alpha} \leq \rho=\sum_{v \in V \setminus \{u\}}e^{-d_R(u,v)\alpha} \leq 2\sum_{i=1}^{+\infty}e^{i\alpha}$. 
Therefore, $\rho=\Theta(1)$. 
A vertex $u$ on the leaves of $RT(k)$ 
	has a long-range edge with ring distance greater than $k$ with 
	probability $\Theta(e^{-k\alpha})/\rho$. 
Let $k=\frac{2}{\alpha}\log n$, we know that a vertex 
	has a long-range edge of ring distance greater 
	than $\frac{2}{\alpha}\log n$ with probability $\Theta(1/n^2)=o(1/n)$.

Therefore, with probability $1-o(1)$, 
	long-range edges never exceed ring distance $\frac{2}{\alpha}\log n$. 
From Theorem \ref{thm:random-ringed-tree-delta} 
	it follows that
	$\delta(RRT(k,e^{-\alpha d_R(u,v)}))=O(\log \log n)$, for any $\alpha > 0$.
\end{proof}




For the case of $\delta(RRT(k,d_R(u,v)^{-\alpha}))$, 
	we first look at a lemma about the effect of long-range edges
	with ring distance $\Omega(n^c)$ 
	on $\delta$-hyperbolicity of ringed trees.

\begin{lem} \label{lem:large-jumping}
If we add an edge between $u$ and $v$ 
	on the outermost ring of a $k$ level ringed tree with 
	$d_R(u,v) \geq c'n^{c}$, for some constants $c$ and $c'$,
	then the resulted graph $G$ (possibly with other edges on the outermost 
	ring) has $\delta(G) = \Omega(\log n)$.
\end{lem}

\begin{proof}
Let $w$ be a node of lowest layer number on $\langle u,v \rangle$. 
By Lemma \ref{lem:dist-estm-by-ring-dist} and the structure of canonical geodesic, 
	$d(w,u) \geq (2\log_2 c'n^c - 3)/2 \geq c\log_2 n + \log_2 c' - 3/2$. We consider the midpoint $x$ of $[u,w]$. For any point $y$ in $[v,w]$, 
	by considering the canonical geodesic 
	$\langle x,y \rangle$, we know that $d(x,y) \geq (c\log_2 n + \log_2 c' - 3/2) / 2 - 3$. 
Therefore $\Delta(u,v,w)$ is at best $((c\log_2 n + \log_2 c' - 3/2) / 2 - 3)$-slim, and thus $\delta(G) = \delta_{Rips}(G) = \Omega(\log n)$.
\end{proof}

A probabilistic version comes naturally as the following corollary.

\begin{coro} \label{coro:large-jumping}
For a random graph $G$ formed by linking edges on leaves of a ringed tree.
	if for some constant $c$ with $0 < c < 1$, 
	with high probability there exists an edge linking some $u$ and 
	$v$ with $d_R(u,v)=\Theta(n^c)$, 
	then with high probability $\delta(G)=\Theta(\log n)$.
\end{coro}

\begin{proof}
Diameter of ringed tree gives $O(\log n)$ upper bound. Lemma \ref{lem:large-jumping} gives $\Omega(\log n)$ lower bound.
\end{proof}

We can now estimate the $\delta$-hyperbolicity of $RRT(k,d_R(u,v)^{-\alpha})$.

\begin{proof}[Proof of Theorem \ref{thm:randomRT}, part 2]
Note that in a ringed tree with $n$ vertices, at least
	$n/2$ of them are leaves.
For a constant $c$ with $0< c < 1$ and
	a fixed vertex $u$, the probability that the 
	long-range edge $(u,v)$ has $d_R(u,v) \leq (n/4)^c$ 
	(we say that it is good) is $p = 2\rho^{-1}\sum_{d=1}^{(n/4)^c}d^{-\alpha}$,
	where $\rho = \Theta(\sum_{d=1}^n d^{-\alpha})$.

For $0 < \alpha < 1$, $p=O(n^{(1-\alpha)(c-1)})=o(1)$. For $\alpha = 1$, 
	$p=c + o(1)$. 
In these two cases, all edges are good with probability at most 
	$(c + o(1))^{(n/2)}=o(1)$. 

For $\alpha > 1$, $\rho = O(1)$. We take $q=1-p$. We have $q=O(n^{c(1-\alpha)})$. By picking $c = \min(1, \frac{1}{2(\alpha-1)})$, we have $q=O(n^{-1/2})$. All edges are good with probability $(1-q)^{n/2}=O(e^{-\sqrt{n}})=o(1)$.

In all three cases, by Corollary \ref{coro:large-jumping}, with 
	probability $1-o(1)$, $\delta(RRT(k,d_R(u,v)^{-\alpha}))=\Theta(\log n)$. 
\end{proof}

Finally, we estimate the $\delta$-hyperbolicity of $RRT(k,2^{-\alpha h(u,v)})$.


\begin{proof}[Proof of Theorem \ref{thm:randomRT}, part 3]
We note $n_L = (n+1)/2$ the number of leaves.
Fix a leaf $u$. 
There are $2^{h-1}$ leaves $v$ such that $h(u,v)=h$. Therefore $\rho = \sum_{h=1}^{\log_2 n_L} 2^{h - 1 - \alpha h}$. For a constant $c$ with
	$0<c < 1$, let $p(c)=\rho^{-1}\sum_{h=1}^{c\log_2 n_L} 2^{h - 1 - \alpha h}$ be the probability that $u$ never links to any $v$ with $h(u,v) \geq c\log_2 n$.
We have $\rho = 2^{-\alpha}n_L^{1-\alpha} / (1 - 2^{1-\alpha})$ for $\alpha \neq 1$, and $\rho = \frac{1}{2}\log_2 n_L$ for $\alpha=1$. For $\alpha \neq 1$, $p(c)=n_L^{-(1-\alpha)(1-c)}$. For $\alpha = 1$, $p(c)=c$.

For the case $\alpha \leq 1$, $p(1/2)=O(1)$. Therefore with probability $(1-p(1/2))^{n_L}=o(1)$ there exists some $u,v$ linked together with $h(u,v) \geq \frac{1}{2}\log_2 n$. For $\alpha > 1$, we take constant $c_0 = \min(1, \frac{1}{2(\alpha-1)})$, and $p(c_0)=n_L^{-1/2}$. Therefore with probability $(1-p(c_0))^{n_L}=O(e^{-\sqrt{n_L}})=o(1)$ linked together with $h(u,v) \geq c_0\log_2 n_L$. In any cases, with $1-o(1)$ probability, there exists $u,v$ linked together by long-range edge with $h(u,v) \geq c\log_2 n_L$ for some constant $c$. We notice that this occurs uniformly through all edges.

Given $u,v$ with $h(u,v)=h$, $d_R(u,v) < 2^{h/2}$ with probability at most $2^{h}2^{1-2h}=2^{1-h}$ by simply counting pairs within ring distance $2^{h/2}$. With $h(u,v) \geq c\log_2 n_L, c>0$, we know that with probability $1 - 2n_L^{-c}=1-o(1)$ we have $d_R(u,v) \geq 2^{h/2} = n_L^{c/2}$. Combining with the previous analysis, we prove that with $1-o(1)$ probability, there exists $u,v$ linked together by long-range edge with $d_R(u,v) \geq n_L^{c/2}$ for some constant $c>0$. By Corollary \ref{coro:large-jumping} and by $n_L > n/2$, with probability $1-o(1)$ we have $\delta(RRT(k,2^{-\alpha h(u,v)}))=\Theta(\log n)$.
\end{proof}




\subsubsection{Proof of Theorem~\ref{thm:randomRBT}}

We can order a binary tree to give it a ring distance. We will suppose that such a distance is defined hereinafter. We begin with a counter part of Lemma \ref{lem:large-jumping} in binary tree.

\begin{lem}\label{lem:large-jump-binary}
If there is an edge between two leaves $u,v$ of a binary tree of size $n$ with distance to lowest common ancestor $h(u,v)=c_1 \log_2 n + c_2$ for some constant $c_1>0, c_2>0$, then the resulted graph $G$ (possibly with other edges on the outermost ring) has $\delta(G) = \Omega(\log n)$.
\end{lem}

\begin{proof}
Consider $w$ the lowest common ancestor of $u,v$, and $x$ the midpoint of $[u,w]$. We have $d(x,w)=d(w,u)/2=h(u,v)/2$. For any $y$ in $[w,v]$, $d(x,y) \geq h(u,v)/2$, as the only path in the tree from $y$ to $x$ always passes by $w$, and we need to climb $h(u,v)/2$ levels if we use links on leaves. Therefore $\Delta(u,v,w)$ is at best $h(u,v)/2=\frac{1}{2}(c_1 \log_2 n + c_2)$, and we conclude that $\delta(G) = \Omega(\log n)$.
\end{proof}

\begin{coro}\label{coro:large-jump-binary}
For a random graph $G$ formed by linking edges on leaves of a binary tree.
	if for some constant $c$ with $0 < c < 1$, 
	with high probability there exists an edge linking some $u$ and 
	$v$ with $h(u,v)=\Theta(\log n)$, 
	then with high probability $\delta(G)=\Theta(\log n)$.
\end{coro}

\begin{proof}
Diameter of binary tree gives $O(\log n)$ upper bound. Lemma \ref{lem:large-jump-binary} gives $\Omega(\log n)$ lower bound.
\end{proof}

\begin{proof}[Proof of Theorem \ref{thm:randomRBT}]
For $RBT(k,e^{-\alpha d_R(u,v)})$, the height of the whole tree is 
	$h = \lfloor \log_2 n \rfloor$. 
There are $\Theta(\sqrt{n})$ subtrees of height $h/2$, with the root
	at the level $h/2$. 
For every neighboring such subtrees, the rightmost leaf $u$ on the 
	left subtree and the leftmost leaf $v$ on the right subtree 
	verifies $d_R(u,v)=1$, $h(u,v) \geq h/2$. For each leaf, $\rho=O(1)$. 
Therefore, for $u,v$ with $d_R(u,v)=1$, there is an extra edge between $u,v$ 
	with constant probability $e^{-\alpha}\rho^{-1} > 0$. 
As there are $\sqrt{n}$ such pairs, with probability 
	$1-(1-e^{-\alpha}\rho^{-1})^{\sqrt{n}}=1-o(1)$, 
	there is a pair of leaves $u,v$ linked by an extra edge with 
	$d_R(u,v)=1$, $h(u,v) \geq h/2 = \Theta(\log n)$. 
By Corollary \ref{coro:large-jump-binary}, with probability $1-o(1)$, 
	$\delta(RBT(k,e^{-\alpha d_R(u,v)}))=\Theta(\log n)$.

For $RBT(k,d_R(u,v)^{-\alpha})$ and $RBT(k,2^{-\alpha h(u,v)})$, using the same analysis in the proof of Theorem \ref{thm:randomRT}, we know that for some constant $c>0$, with $1-o(1)$ probability, there is an extra edge between $u,v$ with $d_R(u,v) = \Omega(n^c)$. We have $h(u,v)=\Omega(\log n)$ because a subtree of height $h$ spans a ring distance at most $2^h$. By Corollary \ref{coro:large-jump-binary}, we have $\delta(RBT(k,d_R(u,v)^{-\alpha}))=\delta(RBT(k,2^{-\alpha h(u,v)}))=\Theta(\log n)$.
\end{proof}

\subsection{Extensions of random ringed tree model} \label{sec:rtringextension}

We will now discuss some extensions of the random ringed tree (RRT) model,
	and show that our results still hold for these extensions,
	thus extending its expressivity.

We start from some observations in the proof of Theorem \ref{thm:randomRT}. 
In this proof, the upper bound of $\delta$-hyperbolicity is given by 
	Theorem \ref{thm:random-ringed-tree-delta}, and 
	the lower bound is given by Corollary \ref{coro:large-jumping}. 
In the statement of Theorem \ref{thm:random-ringed-tree-delta}, by the
	definition of $RT(k,f)$, only a uniform bound of ring distance
	$d_R(u,v)$ for each long-range edge $(u,v)$ is considered. 
In the statement of Corollary \ref{coro:large-jumping}, the only quantity
	concerning a long-range edge $(u,v)$ is also the ringed distance
	between $u$ and $v$, and to apply this corollary, we only need to show that
	a long-range edge $(u,v)$ with $d_R(u,v)=\Theta(n^c)$ for some
	constant $c$ exists with high probability. 
Therefore, the proof of
	Theorem \ref{thm:randomRT} relies only on the ring distances of
	long-range edges.

To extend the RRT model while keeping similar properties on
$\delta$-hyperbolicity, we only need to show that Theorem
\ref{thm:random-ringed-tree-delta} and Corollary
\ref{coro:large-jumping} are still applicable in these extensions. We
will here discuss two extensions on choosing long-range edges.

\vspace{\topsep}

\noindent{\bf A constant number of long-range edges for each node.} In the
original RRT model, each node only have one long-range edge
	connecting to other nodes. 
We can extend the model to allow each node to have a constant
	number of long-range edges connecting to a constant number of
	other nodes.  
In this extension, Theorem \ref{thm:randomRT}
still holds, since Theorem \ref{thm:random-ringed-tree-delta} is not
concerned by the number of long-range edges, and Corollary
\ref{coro:large-jumping} is still applicable as the required
probability only increases with extra long-range edges.

\vspace{\topsep}

\noindent{\bf Independent long-range edges.} In the original RRT model, we
choose exactly one long-range edge for each node. 
A variant of the model is that each node $u$ can choose edge $(u,v)$
	independently from other edges $(u,v')$, with the same probability as
	in the original model such that on expectation $u$ connect out
	with one long-range edge.
In this variant, Theorem \ref{thm:randomRT}
still holds. The reason is that in expectation, at least a constant fraction of nodes issue only one edge, and the
computation for applying Corollary \ref{coro:large-jumping} stays similar. It is clear that the application of Theorem
\ref{thm:random-ringed-tree-delta} stays valid.

\vspace{\topsep}

In the two variants discussed above, we can see that Theorem
\ref{thm:randomRT} still applies, and we have exactly the same
property on $\delta$-hyperbolicity of these variants. We can also
combine these two variants, and it is clear that our results are still
valid.

\section{Discussions and open problems} \label{sec:discuss}

Perhaps the most obvious extension of our results is to close the gap in 
the bounds on the hyperbolicity in the low-dimensional small-world model when 
$\gamma$ is at the ``sweetspot,'' as well as extending the results for 
large $\gamma$ to dimensions $d\ge 2$.  
Also of interest is characterizing in more detail the hyperbolicity 
properties of other random graph models, in particular those that have 
substantial heavy-tailed properties.
Finally, exact computation of $\delta$ by its definition takes $O(n^4)$ time, which is not scalable
	to large graphs, and thus the design of more efficient exact or approximation algorithms
	would be of interest.

From a broader perspective, however, our results suggest that $\delta$ is a 
measure of tree-like-ness that can be quite sensitive to noise in graphs, 
and in particular to randomness as it is implemented in common network 
generative models.  
For example, the ringed trees have constant hyperbolicity but once adding some random links
	among leaves, our results show that very likely their hyperbolic $\delta$ 
	reaches the level of graph diameter and they become	not hyperbolic at all.
Moreover, our results for the $\delta$ hyperbolicity of rewired trees 
(Theorem~\ref{thm:randomRBT})
versus 
rewired low-$\delta$ tree-like metrics 
(Theorem~\ref{thm:randomRT}~(1)) suggest that, while quite appropriate 
for continuous negatively-curved manifolds, the usual definition of $\delta$ 
may be somewhat less useful for discrete graphs.  
Thus, it would be of interest to address questions such as: does there exist 
a measure other than Gromov's $\delta$ that is more appropriate for 
graph-based data or more robust to noise/randomness as it is used in popular 
network generation models; is it possible to incorporate in a meaningful way 
nontrivial randomness in other low $\delta$-hyperbolicity graph families; and can 
we construct non-trivial random graph families that contain as much randomness as possible while
having low $\delta$-hyperbolicity comparing to graph diameter?

\section*{Acknowledgments}

We are grateful to Yajun Wang and
	Xiaohui Bei for their helpful discussions on this topic.

\bibliography{hyperbolic,communities,mwmbib_jrnl,mwmbib_proc,mwmbib_book,mwmbib_misc}

\begin{thebibliography}{10}

\bibitem{ABKMRT07}
I.~Abraham, M.~Balakrishnan, F.~Kuhn, D.~Malkhi, V.~Ramasubramanian, and
  K.~Talwar.
\newblock Reconstructing approximate tree metrics.
\newblock In {\em Proceedings of the 26th Annual ACM Symposium on Principles of
  Distributed Computing}, pages 43--52, 2007.

\bibitem{Bar02}
Y.~Baryshnikov.
\newblock On the curvature of the {I}nternet.
\newblock in Workshop on Stochastic Geometry and Teletraffic, Eindhoven, The
  Netherlands, April 2002.

\bibitem{BT10a_TR}
Y.~Baryshnikov and G.~H. Tucci.
\newblock Asymptotic traffic flow in an hyperbolic network {I}: Definition and
  properties of the core.
\newblock Technical Report Preprint: arXiv:1010.3304, 2010.

\bibitem{BRSV10}
S.~Bermudo, J.~M. Rodr\'{i}guez, J.~M. Sigarreta, and J.-M. Vilaire.
\newblock Mathematical properties of {G}romov hyperbolic graphs.
\newblock In {\em Proceedings of the 2010 International Conference of Numerical
  Analysis and Applied Mathematics}, pages 575--578, 2010.

\bibitem{BKC09}
M.~{Bogu{\~{n}}\'{a}}, D.~Krioukov, and K.~Claffy.
\newblock Navigability of complex networks.
\newblock {\em Nature Physics}, 5:74--80, 2009.

\bibitem{bowditch1991notes}
B.~H. Bowditch.
\newblock Notes on {G}romov's hyperbolicity criterion for path-metric spaces.
\newblock {\em Group theory from a geometrical viewpoint (Trieste, 1990)},
  pages 64--167, 1991.

\bibitem{BH99}
M.~R. Bridson and A.~Haefliger.
\newblock {\em Metric Spaces of Non-Positive Curvature}.
\newblock Springer, 1999.

\bibitem{CD00}
V.~Chepoi and F.~Dragan.
\newblock A note on distance approximating trees in graphs.
\newblock {\em European Journal of Combinatorics}, 21(6):761--766, 2000.

\bibitem{chepoi2008diameters}
V.~Chepoi, F.~F. Dragan, B.~Estellon, M.~Habib, and Y.~Vax{\`e}s.
\newblock Diameters, centers, and approximating trees of $\delta$-hyperbolic
  geodesic spaces and graphs.
\newblock In {\em Proceedings of the 24th Annual Symposium on Computational
  Geometry}, pages 59--68. ACM, 2008.

\bibitem{CDEHVX10}
V.~Chepoi, F.~F. Dragan, B.~Estellon, M.~Habib, Y.~Vax\`{e}s, and Y.~Xiang.
\newblock Additive spanners and distance and routing labeling schemes for
  hyperbolic graphs.
\newblock {\em Algorithmica}, 62(3-4):713--732, 2012.

\bibitem{CE07}
V.~Chepoi and B.~Estellon.
\newblock Packing and covering $\delta$-hyperbolic spaces by balls.
\newblock In {\em Proceedings of the 10th International Workshop on
  Approximation}, pages 59--73, 2007.

\bibitem{montgolfier2011treewidth}
F.~de~Montgolfier, M.~Soto, and L.~Viennot.
\newblock Treewidth and hyperbolicity of the internet.
\newblock In {\em IEEE Networks Computing and Applications 2011}. IEEE, 2011.

\bibitem{GL05}
C.~Gavoille and O.~Ly.
\newblock Distance labeling in hyperbolic graphs.
\newblock In {\em Proceedings of the 16th Annual International Symposium on
  Algorithms and Computation}, pages 1071--1079, 2005.

\bibitem{ghys1990groupes}
{\'E}.~Ghys and P.~de~La~Harpe.
\newblock {\em Sur Les Groupes Hyperboliques D'apr{\`e}s Mikhael Gromov}.
\newblock Birkh{\"a}user, 1990.

\bibitem{gromov1987hyperbolic}
M.~Gromov.
\newblock Hyperbolic groups.
\newblock {\em Essays in group theory}, 8:75--263, 1987.

\bibitem{Ho79}
E.~Howorka.
\newblock On metric properties of certain clique graphs.
\newblock {\em Journal of Combinatorial Theory, Series B}, 27(1):67--74, 1979.

\bibitem{JL02}
E.~Jonckheere and P.~Lohsoonthorn.
\newblock Hyperbolic geometry approach to multipath routing.
\newblock In {\em Proceedings of the 10th Mediterranean Conference on Control
  and Automation}, 2002.

\bibitem{JL04}
E.~Jonckheere and P.~Lohsoonthorn.
\newblock Geometry of network security.
\newblock In {\em Proceedings of the 2004 American Control Conference}, pages
  2:976--981, 2004.

\bibitem{JLB08}
E.~Jonckheere, P.~Lohsoonthorn, and F.~Bonahon.
\newblock Scaled {G}romov hyperbolic graphs.
\newblock {\em Journal of Graph Theory}, 57(2):157--180, 2008.

\bibitem{JLBB11}
E.~Jonckheere, M.~Lou, F.~Bonahon, and Y.~Baryshnikov.
\newblock Euclidean versus hyperbolic congestion in idealized versus
  experimental networks.
\newblock {\em Internet Mathematics}, 7(1):1--27, 2011.

\bibitem{JLHB07}
E.~A. Jonckheere, M.~Lou, J.~Hespanha, and P.~Barooah.
\newblock Effective resistance of {G}romov-hyperbolic graphs: Application to
  asymptotic sensor network problems.
\newblock In {\em Proceedings of the 46th IEEE Conference on Decision and
  Control}, pages 1453--1458, 2007.

\bibitem{Kle00}
J.~Kleinberg.
\newblock The small-world phenomenon: an algorithm perspective.
\newblock In {\em Proceedings of the 32nd Annual ACM Symposium on Theory of
  Computing}, pages 163--170, 2000.

\bibitem{Kle01}
J.~Kleinberg.
\newblock Small-world phenomena and the dynamics of information.
\newblock In {\em Proceedings of Annual Advances in Neural Information
  Processing Systems}, pages 431--438, 2002.

\bibitem{Kle06}
J.~Kleinberg.
\newblock Complex networks and decentralized search algorithms.
\newblock In {\em Proceedings of the International Congress of Mathematicians},
  2006.

\bibitem{Kle07}
R.~Kleinberg.
\newblock Geographic routing using hyperbolic space.
\newblock In {\em Proceedings of the 26th IEEE International Conference on
  Computer Communications}, pages 1902--1909, 2007.

\bibitem{KCFB07}
D.~Krioukov, k.~c. claffy, K.~Fall, and A.~Brady.
\newblock On compact routing for the {I}nternet.
\newblock {\em Computer Communication Review}, 37(3):41--52, 2007.

\bibitem{KPKVB10}
D.~Krioukov, F.~Papadopoulos, M.~Kitsak, A.~Vahdat, and M.~{Bogu{\~{n}}\'{a}}.
\newblock Hyperbolic geometry of complex networks.
\newblock {\em Physical Review E}, 82:036106, 2010.

\bibitem{KPVB09}
D.~Krioukov, F.~Papadopoulos, A.~Vahdat, and M.~{Bogu{\~{n}}\'{a}}.
\newblock Curvature and temperature of complex networks.
\newblock {\em Physical Review E}, 80:035101(R), 2009.

\bibitem{LR94}
J.~Lamping and R.~Rao.
\newblock Laying out and visualizing large trees using a hyperbolic space.
\newblock In {\em Proceedings of the 7th Annual ACM Symposium on User Interface
  Software and Technology}, pages 13--14, 1994.

\bibitem{LRP95}
J.~Lamping, R.~Rao, and P.~Pirolli.
\newblock A focus+context technique based on hyperbolic geometry for
  visualizing large hierarchies.
\newblock In {\em Proceedings of CHI'95: Conference on Human Factors in
  Computing}, pages 401--408, 1995.

\bibitem{LohsoonTH}
P.~Lohsoonthorn.
\newblock {\em Hyperbolic Geometry of Networks}.
\newblock PhD thesis, University of. Southern California, 2003.

\bibitem{LouTH}
M.~Lou.
\newblock {\em Traffic pattern in negatively curved network}.
\newblock PhD thesis, University of Southern California, 2008.

\bibitem{MN04}
C.~U. Martel and V.~Nguyen.
\newblock Analyzing {K}leinberg's (and other) small-world models.
\newblock In {\em Proceedings of the 23rd Annual ACM Symposium on Principles of
  Distributed Computing}, pages 179--188, 2004.

\bibitem{Mun98}
T.~Munzner.
\newblock Exploring large graphs in 3{D} hyperbolic space.
\newblock {\em IEEE Computer Graphics and Applications}, 18(4):18--23, 1998.

\bibitem{MB95}
T.~Munzner and P.~Burchard.
\newblock Visualizing the structure of the {W}orld {W}ide {W}eb in 3{D}
  hyperbolic space.
\newblock In {\em Proceedings of the First Symposium on Virtual Reality
  Modeling Language}, pages 33--38, 1995.

\bibitem{NS11}
O.~Narayan and I.~Saniee.
\newblock Large-scale curvature of networks.
\newblock {\em Phys. Rev. E}, 84:066108, Dec 2011.

\bibitem{NST12}
O.~Narayan, I.~Saniee, and G.~Tucci.
\newblock Lack of spectral gap and hyperbolicity in asymptotic
  {Erd\H{o}s-R\'{e}nyi} random graphs.
\newblock In {\em Proceedings of the 5th International Symposium on
  Communications Control and Signal Processing (ISCCSP)}, pages 1--4, 2012.

\bibitem{NW99}
M.~E.~J. Newman and D.~J. Watts.
\newblock Scaling and percolation in the small-world network model.
\newblock {\em Physical Review E}, 60(6):7332--7342, 1999.

\bibitem{NM05}
V.~Nguyen and C.~U. Martel.
\newblock Analyzing and characterizing small-world graphs.
\newblock In {\em Proceedings of the 16th Annual ACM-SIAM Symposium on Discrete
  Algorithms}, pages 311--320, 2005.

\bibitem{PKBV10}
F.~Papadopoulos, D.~V. Krioukov, M.~Bogu{\~n}{\'a}, and A.~Vahdat.
\newblock Greedy forwarding in dynamic scale-free networks embedded in
  hyperbolic metric spaces.
\newblock In {\em Proceedings of the 29th IEEE International Conference on
  Computer Communications (INFOCOM)}, pages 2973--2981, 2010.

\bibitem{PR00}
C.~Papazian and E.~R\'{e}mila.
\newblock Some properties of hyperbolic networks.
\newblock In {\em Proceedings of the 9th International Conference on Discrete
  Geometry for Computer Imagery}, pages 149--158, 2000.

\bibitem{Shang12}
Y.~Shang.
\newblock Lack of {G}romov-hyperbolicity in small-world networks.
\newblock {\em Central European Journal of Mathematics}, 10(3):1152--1158,
  2012.

\bibitem{ST08}
Y.~Shavitt and T.~Tankel.
\newblock Hyperbolic embedding of {I}nternet graph for distance estimation and
  overlay construction.
\newblock {\em IEEE/ACM Transactions on Networking}, 16(1):25--36, 2008.

\bibitem{Tucci12}
G.~H. Tucci.
\newblock Random regular graphs are not asymptotically {G}romov hyperbolic.
\newblock Technical Report Preprint: arXiv:1203.5069 (2012), 2012.

\bibitem{WR02}
J.~A. Walter and H.~Ritter.
\newblock On interactive visualization of high-dimensional data using the
  hyperbolic plane.
\newblock In {\em Proceedings of the 8th Annual ACM SIGKDD Conference}, pages
  123--132, 2002.

\end{thebibliography}
\bibliographystyle{abbrv}

%



\end{document}